\documentclass[11pt]{article}

\def\focs{0} 

\usepackage[T1]{fontenc}
\usepackage{ae,aecompl}

\usepackage[colorlinks=true,citecolor=blue,linkcolor=black,urlcolor=green]{hyperref}
\usepackage{url, fullpage}

\usepackage{amsfonts, amsmath, amssymb}
\usepackage{color}

\def\mnotes{0} 

\ifnum\mnotes=0
\newcommand{\mnote}[1]{}
\else
	\newcounter{mynotes}
	\newcommand{\mnote}[1]{\addtocounter{mynotes}{1}{{\bf !}}
	\marginpar{\scriptsize  {\arabic{mynotes}.\ {\sf \textcolor{red}{#1}}}}}
	\fi

\newcommand{\enote}[1]{\mnote{E: #1}}
\newcommand{\anote}[1]{\mnote{A: #1}}

\newtheorem{theorem}{Theorem}
\newtheorem{corollary}[theorem]{Corollary}
\newtheorem{lemma}[theorem]{Lemma}
\newtheorem{observation}[theorem]{Observation}
\newtheorem{proposition}[theorem]{Proposition}
\newtheorem{definition}[theorem]{Definition}
\newtheorem{claim}[theorem]{Claim}

\newtheorem{conjecture}[theorem]{Conjecture}

\ifnum\focs=0
\newcommand{\qed}{\rule{7pt}{7pt}}
\newenvironment{proof}{\noindent{\bf Proof:}\hspace*{1em}}{\qed\bigskip}
\newenvironment{proof-sketch}{\noindent{\bf Proof Sketch:}\hspace*
{1em}}{\qed\bigskip}
\newenvironment{proof-idea}{\noindent{\bf Proof Idea:}\hspace*{1em}}
{\qed\bigskip}
\newenvironment{proof-of-lemma}[1]{\noindent{\bf Proof of Lemma #1}
  \hspace*{1em}}{\qed\bigskip}
\newenvironment{proof-attempt}{\noindent{\bf Proof Attempt}\hspace*
{1em}}{\qed\bigskip}
\newenvironment{proofof}[1]{\noindent{\bf Proof
of #1:}}{\qed\bigskip}
\newenvironment{remark}{\noindent{\bf Remark:}\hspace*{1em}}{\bigskip}
\else
\newcommand{\qed}{\IEEEQED}
\newenvironment{proof-sketch}{\begin{IEEEproof}[Sketch of Proof]}{\end{IEEEproof}}
\newenvironment{proof-idea}{\begin{IEEEproof}[Proof Idea]}{\end{IEEEproof}}
\newenvironment{proof-of-lemma}[1]{\begin{IEEEproof}[Proof of Lemma #1]}{\end{IEEEproof}}
\newenvironment{proof-attempt}{\begin{IEEEproof}[Proof Attempt]}{\end{IEEEproof}}
\newenvironment{proofof}[1]{\begin{IEEEproof}[Proof of #1]}{\end{IEEEproof}}
\newenvironment{remark}{\begin{IEEEproof}[Remark]}{\end{IEEEproof}}
\fi

\newcommand{\calf}{\mathcal{F}}
\newcommand{\calp}{\mathcal{P}}
\newcommand{\calh}{\mathcal{H}}
\newcommand{\cals}{\mathcal{S}}
\newcommand{\calm}{\mathcal{M}}
\newcommand{\eps}{\epsilon}

\DeclareMathOperator*{\E}{\mathbb{E}}

\newcommand{\Z}{{\mathbb Z}}
\newcommand{\F}{{\mathbb F}}

\newcommand{\N}{{\mathbb N}}

\newcommand{\eqdef}{{\stackrel{\rm def}{=}}}

\newcommand{\ind}{\mathop{\mathsf{ind}}}

\newcommand{\mspan}{{\rm span}}
\newcommand{\zo}{{\{0,1\}}}

\newcommand{\ignore}[1]{}

\newcommand{\indic}{\mathbf{1}}

\begin{document}

\title{A Unified Framework for Testing Linear-Invariant Properties}

\title{A Unified Framework for Testing Linear-Invariant Properties}
\author{Arnab Bhattacharyya\thanks{Computer Science and Artificial
  Intelligence Laboratory, MIT.  Email: {\tt abhatt@mit.edu}.  Supported in part by a DOE
  Computational Science Graduate Fellowship and NSF
Awards 0514771, 0728645, and 0732334.} \and Elena Grigorescu\thanks{Computer Science and Artificial
  Intelligence Laboratory, MIT.  Email: {\tt elena\_g@csail.mit.edu}. Supported by NSF award CCR-0829672.}
\and Asaf
Shapira\thanks{School of Mathematics and School of Computer Science, Georgia Institute of
  Technology, Atlanta, GA 30332. Email: {\tt asafico@math.gatech.edu}.  Supported in part by NSF
  Grant DMS-0901355.}}

%
%
%

\maketitle

\begin{abstract}

 The study of the interplay between the testability of properties of Boolean functions and the invariances acting on their domain which preserve the property was initiated by  Kaufman and Sudan (STOC 2008).
Invariance with respect to $\F_2$-linear transformations is arguably the most common symmetry exhibited by natural properties of Boolean functions on the hypercube. Hence, an important goal in Property Testing is to describe necessary and
sufficient conditions for the testability of linear-invariant properties. This direction was explicitly proposed for investigation in a recent
survey of Sudan.
 We obtain the following results:

\begin{enumerate}
\item We show that every linear-invariant property that can be characterized
by forbidding induced solutions to a (possibly infinite) set of linear equations can be tested with
one-sided error.

\item We show that every linear-invariant property that can be tested with one-sided error can be
  characterized by forbidding induced solutions to a (possibly infinite) set of {\em systems} of linear equations.
\end{enumerate}

We conjecture that our result from item (1) can be extended to cover systems of linear equations. We further show that the validity of this conjecture
would have the following implications:

\begin{enumerate}
\item It would imply that every linear-invariant property that
is closed under restrictions to linear subspaces is testable with one-sided error. Such a result would unify several previous results on testing
Boolean functions, such as the testability of low-degree polynomials and of Fourier dimensionality.

\item It would imply that a linear-invariant property ${\cal P}$ is testable with one-sided error {\bf if and only if} ${\cal P}$ is closed under restrictions to linear subspaces, thus
resolving Sudan's problem.
\end{enumerate}

\end{abstract}

\thispagestyle{empty}
\newpage
\setcounter{page}{1}

\section{Introduction}\label{sec:intro}

Let ${\cal P}$ be a property of Boolean functions. A {\em testing} algorithm for ${\cal P}$ is a randomized algorithm that can quickly
distinguish between the case that $f$ satisfies ${\cal P}$ from the case that $f$ is far from satisfying ${\cal P}$. The problem of characterizing the
properties of Boolean functions for which such an efficient algorithm exists is considered by many to be the most important open problem in this area. Since a complete characterization
seems to be out of reach, several researchers have recently considered the problem of characterizing the testable properties ${\cal P}$ that belong to certain
``natural'' subfamilies of properties. One such family that has been extensively studied is the
family of so called {\em linear-invariant} properties. Our main result is two fold. We first show
that every property in a large family of linear-invariant properties is indeed testable. Next, we
conjecture that an even more general family of properties can be tested and show that such a result 
would give a {\em characterization} of the linear-invariant properties that are testable with
one-sided error.

\subsection{Background on property testing}\label{subsecbackground}

We start with the formal definitions related to testing Boolean functions. Let ${\cal P}$ be a property of Boolean functions over the $n$-dimensional Boolean hypercube. In other words, ${\cal P}$ is simply a subset of
the set of functions $f: \{0,1\}^n \to \{0,1\}$. Two functions $f,g : \{0,1\}^n \to \{0,1\}$ are $\epsilon$-far if they differ on at least $\epsilon 2^n$ of the inputs.
We say that $f$ is $\epsilon$-far from satisfying a property ${\cal P}$ if it is $\epsilon$-far from any function $g$ satisfying ${\cal P}$. A {\em tester} for the property ${\cal P}$ is a randomized algorithm
which can quickly distinguish between the case that an input function $f$ satisfies ${\cal P}$ from the case that it is $\epsilon$-far from satisfying ${\cal P}$. Here we assume that
the input function $f$ is given to the tester as an oracle, that is, the tester can ask an oracle for the value of the input functions $f$ on a certain $x \in \{0,1\}^n$.
We say that ${\cal P}$ is {\em strongly testable} (or simply {\em testable}) if ${\cal P}$ has a
tester which makes only a constant number of queries to the oracle, where this constant can depend on
$\epsilon$ but should be independent\footnote{Observe that since we aim for asymptotic results (that is, we think of $n \rightarrow \infty$), our property ${\cal P}$ can actually be described as ${\cal P}=\bigcup^{\infty}_{i=1}{\cal P}_n$, where ${\cal P}_n$ is the collection of functions over the $n$-dimensional Boolean hypercube which satisfy ${\cal P}$.} of $n$. Finally, we say that a testing
algorithm has {\em one-sided} error if it always accepts input functions satisfying ${\cal P}$. (We
always demand that the tester rejects input functions which are $\eps$-far from satisfying ${\cal P}$ with probability at
least, say, $2/3$.)

The study of testing of Boolean functions began with the work of Blum, Luby and Rubinfeld \cite{BLR} on testing linearity of Boolean functions. This work was further extended by Rubinfeld and Sudan \cite{RS}.
Around the same time, Babai, Fortnow and Lund \cite{BFL} also studied similar problems as part of their work on MIP=NEXP. These works are all related to the PCP Theorem, and an important part of it involves tasks
which are similar in nature to testing properties of Boolean functions. The work of Goldreich, Goldwasser and Ron \cite{GGR} extended these results to more combinatorial settings, and initiated the study of
similar problems in various areas. More recently, numerous testing questions in the Boolean
functions settings have sparked great interest: testing  dictators \cite{ParnasRS02}, low-degree
polynomials~\cite{AKKLR,Samorodnitsky07}, juntas \cite{FKRSS04,blais-junta}, concise
representations~\cite{diakonikolasLMORSW07}, halfspaces \cite{MatulefORS09}, codes
\cite{KaufmanLitsyn, KS07,
  KoppartyS09}.   
 These are documented in several surveys
\cite{FischerSurvey,RubinfeldICM,RonSurvey,SudanSurvey}, and we refer the reader to these surveys for more background and references on property testing.

\subsection{Invariance in testing Boolean functions}\label{subsecinvariahnce}

What features of a property make it testable?
One area in which this question is
relatively well understood is testing properties of dense graphs \cite{AS08,AFNS06, BCLSSV06}. In sharp
contrast, this question is far from being well understood in the case of testing properties of
Boolean functions. In an attempt to remedy this, Sudan and several coauthors \cite{KS, GKS08, GKS09,
BS09} have recently begun to
investigate the role of invariance in property testing. The idea is that in order to be able to test
if a combinatorial structure satisfies a property using very few queries to its representation, the
property we are trying to test must be closed under certain transformations. For example, when
testing properties of dense graphs, we are allowed to ask if two vertices $i$ and $j$ are adjacent
in the graph, and the assumption is that the property we are testing is invariant under renaming of
the vertices.  In other words, if we think of the input as an ${n \choose 2}$ dimensional $0/1$ vector
encoding the adjacency matrix of the input, then the property should be closed under transformations
(of the edges) which result from permuting the vertices of the graph.

A natural notion of invariance that one can consider when studying Boolean functions over the
hypercube is linear-invariance, which is in some sense the analogue for graph properties being
closed under renaming of the vertices (we further discuss this analogy in Subsection
\ref{subsecmain}). Formally, a property of Boolean functions ${\cal P}$ is said to be
linear-invariant if for every function $f: \F_2^n \to \zo$ satisfying $\cal P$ and for any $\F_2$-linear transformation
$L:\F_2^n
\to \F_2^n$ the function $f \circ L$ satisfies ${\cal P}$ as well, where we define $(f \circ
L)(x)=f(L(x))$. Note that here we identify $\{0,1\}^n$ with $\F_2^n$, and we will use this convention
from now on throughout the paper.  For a thorough discussion of the importance of linear-invariance,
we refer the reader to Sudan's recent survey on the subject \cite{SudanSurvey} and to the paper of
Kaufman and Sudan which initiated this line of work \cite{KS}.

\subsection{The main result}\label{subsecmain}

Our main result in this paper (stated in Theorem \ref{thm:main} below) is that a natural family of linear-invariant properties of Boolean functions can all be tested with one-sided error. The statement requires some preparation.

\begin{definition}[$(M,\sigma)$-free]\label{defmbfree}
Given an $m \times k$ matrix $M$ over $\F_2$ and $\sigma  \in
\zo^k$ for integers $m > 0$ and $k>2$, we say that a function $f: \F_2^n \to \zo$ is 
{\em $(M,\sigma)$-free} if there is no $x = (x_1,\dots,x_k) \in (\F_2^n)^k$ such that $Mx = 0$
and for all $1 \leq i \leq k$ we have $f(x_i)=\sigma_i$.
\ignore{If such an $x$ exists we say that $f$ {\em induces $(M,\sigma)$ at $x$} and denote this by $(M,\sigma) \to f$.}
\end{definition}

\begin{remark}
By removing linearly dependent rows, we can ensure that $\mathsf{rank}(M) = m$
without loss of generality.  We will assume this fact henceforth.
\end{remark}

Let us give some intuition about the above definition. Given a function $f: \F_2^n \to \zo$, it is natural to
consider the set $S_f=\{x \in \F^n_2:f(x)=1\}$. Suppose for the rest of this paragraph that in
the above definition $\sigma=1^k$. In this case $f$ is $(M,\sigma)$-free if and only if $S_f$
contains no solution to the system of equations $Mx=0$, that is, if there is no $v \in S^k_f$
satisfying $Mv=0$. Note that when considering graph properties,
the notion of $(M,1^k)$-freeness is analogous to the graph property of being $H$-free\footnote{If $H$ is a graph on $h$ vertices, then we
say that a graph $G$ is $H$-free if $G$ contains no set of $h$ vertices that contain a copy of $H$ (possibly with some other edges).}, where $H$ is some fixed graph. Observe that in both
cases the property is {\em monotone} in the sense that if $f$ is $(M,1^k)$-free, then removing
elements from $S_f$ results in a set that contains
no solution to $Mx=0$. Similarly if $G$ is $H$-free, then removing edges from $G$ results in an $H$-free graph.

Let us now go back to considering arbitrary $\sigma \in \{0,1\}^k$ in Definition \ref{defmbfree},
where again the intuition comes from graph properties. Observe that a natural variant of the
monotone graph property of being $H$-free is the property of being induced $H$-free\footnote{If $H$
is a graph on $h$ vertices, then we say that a graph $G$ is induced $H$-free if $G$ contains no set
of $h$ vertices that contain a copy of $H$ and no other edges.}. Note that being induced $H$-free is
no longer a monotone property since if $G$ is induced $H$-free then removing an edge can actually
create induced copies of $H$. Getting back to the property of being $(M,\sigma)$-free, observe that
we can think of this as requiring $S_f$ to contain no {\em induced} solution to the system of
equations $Mx=0$. That is, the requirement is that there should be no vector $v$ satisfying $Mv=0$, where $v_i\in S_f$ if
$\sigma_i=1$ and $v_i \in \F^n_2 \setminus S_f$ if $\sigma_i=0$. So we can think of $\sigma$ as encoding which
elements of a potential solution vector $v$ should belong to $S_f$ and which should belong to its
complement. For this reason we will adopt the convention of calling $(M,\sigma)$ a {\em forbidden
induced system of equations}.

Continuing with the graph analogy, once we have the property of being induced $H$-free, for some fixed
graph $H$, it is natural to consider the property of being induced ${\cal H}$-free where ${\cal H}$
is a fixed finite set of graphs.  Several natural graph properties can be described as being induced ${\cal
H}$-free (e.g. being a line-graph), but it is of course natural to further generalize this notion and
allow ${\cal H}$ to contain an infinite number of forbidden induced graphs. One then gets a very
rich family of properties like being Perfect, $k$-colorable, Interval, Chordal etc. This
generalization naturally motivates the following definition which will be key to our main results.

\begin{definition}[${\cal F}$-free]\label{defFfree}
Let $\mathcal{F} = \{(M^1,\sigma^1),$$(M^2,\sigma^2),$$\dots\}$ be a (possibly infinite) set of induced systems of linear equations. A function $f$ is said to be
${\cal F}$-free if it is $(M^i,\sigma^i)$-free\footnote{In the sense of Definition \ref{defmbfree}} for all $i$.
\end{definition}

Observe that this definition is an OR-AND type restriction, that is, we require that $f$ will not satisfy {\em any} of the systems $(M^i,\sigma^i)$, where $f$ satisfies $(M^i,\sigma^i)$ if
it satisfies all the equations of $M^i$ (in the sense of Definition \ref{defmbfree}). We are now ready to state our main result.

\begin{theorem}[Main Result]\label{thm:main}
Let $\mathcal{F} = \{(M^1,\sigma^1),(M^2,\sigma^2),\dots\}$ be a possibly infinite set
of induced equations (that is, all the matrices $M^i$ are of rank one), each on more than two
variables. Then the property of being ${\cal F}$-free is testable 
with one-sided error.
\end{theorem}

Note that, in the above statement, each $M^i$ contains a single equation, rather than a {
system} of equations as in Definition \ref{defFfree}.  In fact, though, what we prove is quite a
bit stronger: Theorem \ref{thm:main} holds when each $M^i$ is of {\em complexity $1$}, instead of
just rank $1$.  The notion of complexity of a linear system is derived from work by Green and Tao
\cite{GT06} (See Section \ref{subsec:extension} for the formal definition.)  There, we also show that
any matrix of rank at most two is of complexity $1$, and, hence, Theorem \ref{thm:main} is obviously
a corollary of this stronger result.  But for the sake of simplicity, let us restrict
ourselves to discussing matrices of rank one in this section.

Let us compare this result to some previous works.  One work that initiated some of the
recent results on testing Boolean functions was obtained by Green \cite{Green05}. His result can be
formulated as saying that for any rank one matrix $M$, the property of being $(M,1^k)$-free can be
tested with one-sided error. Green conjectured that the same result holds for any {\em system} of
linear equations. This conjecture was recently confirmed by Shapira \cite{Shap09} and Kr\'al{'},
Serra and Vena \cite{KSV08b}. In our language, the results of \cite{Shap09,KSV08b} can be stated as
saying that for any matrix $M$, the property of being $(M,1^k)$-free is testable with one-sided
error. The case of arbitrary $\sigma$ was first explicitly considered in \cite{BCSX09} where it
was shown that if $M$ is a rank one matrix, then $(M,\sigma)$-freeness is equivalent to a finite set
of properties, all of which were already known to be testable.  Tim Austin (see \cite{Shap09})
conjectured that the result of \cite{Shap09} for an arbitrary matrix $M$ can be extended to show
testability of  $(M,\sigma)$-freeness for every vector $\sigma$.  Shapira
\cite{Shap09} further conjectured that his result can be extended to the case when we forbid an infinite set of
systems of linear equations as in Definition \ref{defFfree}. So Theorem
\ref{thm:main} partially resolves the above conjecture, since it can handle an infinite number of
induced equations (but not an infinite number of forbidden arbitrary {\em systems} of equations).

Another way to think of Theorem \ref{thm:main} comes (yet again) from the analogy with graph properties. Alon
and Shapira \cite{AS08} have shown that for every set of graphs ${\cal F}$, the property of being
induced ${\cal F}$-free is testable with one-sided error. Since in many ways\footnote{This analogy
  is informal, but see \cite{KSV09} and \cite{SzegSym10} for some formal connections.}, copies of a
fixed graph $H$ in a graph $G$ correspond to finding solutions of a {\em single} equation in a set $S \subseteq
\F^n_2$, Theorem \ref{thm:main} can be considered to be a Boolean functions analog of the result of
\cite{AS08}.  Just like the graph property of being free of a particular subgraph $H$ is analogous to
being $(M,\sigma)$-free where $M$ has rank $1$, the {\em hypergraph} property of being free of a
particular sub-hypergraph ${\cal H}$ is analogous to being  $(M,\sigma)$-free for an arbitrary
$M$. Now, the result of \cite{AS08} has been later extended to hypergraphs by Austin and Tao
\cite{AT08} and R\"odl and Schacht \cite{RodlSchacht}; so, it is natural to expect that one could
also handle an infinite number of forbidden induced systems of equations in the functional case as
well. All the above motivates us to raise the following conjecture.\anote{The Rodl et al. work
  actually only shows testability for hereditary properties of $k$-uniform hypergraphs for every
  {\em fixed} $k$.  So a more reasonable conjecture would be if the rank of each system was
  bounded.  But then in terms of subspace hereditary, this would be awkward to state...}

\begin{conjecture}\label{mainconj}
For every (possibly infinite) set of systems of induced equations ${\cal F}$, the property of being
${\cal F}$-free is testable with one-sided error.
\end{conjecture}

As the reader can easily convince himself, a graph property ${\cal P}$ is equivalent to being induced ${\cal H}$-free if and only if ${\cal P}$ is closed under vertex removal. Such properties are usually called {\em hereditary}.
This motivates us to define the following analogous notion for properties of Boolean functions.

\begin{definition}[Subspace-Hereditary Properties]\label{DefHereditary}
A linear-invariant property $\mathcal{P}$ is said to be {\em subspace-hereditary} if it is
closed under restriction to subspaces. That is, if $f$ is in $\calp_n$ and $H$ is a
$m$-dimensional linear subspace of $\F_2^n$, then $f|_H \in \calp_m$ also, where\footnote{Note that
  we are implicitly composing $f|_H$ with a linear transformation so that it is now defined on
  $\F_2^m$.  Here, we are using the fact that $\calf$ is linear-invariant.}
$f|_H : \F_2^m \to \zo$ is the restriction of $f$ to $H$.
\end{definition}

When considering linear-invariant properties, one can also obtain the following (slightly cleaner)
view of the properties of Definition \ref{defFfree}. This equivalence is analogous to the graph
properties mentioned above. We stress that this equivalence is a further indication of the
``naturalness'' of the notion of linear-invariance and its resemblance to the closure of graph
properties under vertex renaming. We defer its proof to the appendix. 

\begin{proposition}\label{HereditaryInduced} A linear-invariant property ${\cal P}$ is subspace-hereditary
if and only if there is a (possibly infinite) set of systems of induced equations
${\cal F}$ such that ${\cal P}$ is equivalent to being ${\cal F}$-free.
\end{proposition}

We mention that while the notions of graph properties being hereditary and functions being subspace-hereditary
are somewhat more natural than the equivalent notions of being free of induced subgraphs
and equations respectively, it is actually easier to think about these properties using the latter
notion when proving theorems about them. This was the case in \cite{AS08}, and it will be the case
in the present paper as well.  Proposition \ref{HereditaryInduced} along with Conjecture
\ref{mainconj} implies the following:

\begin{corollary}\label{cor:alltestable}
If Conjecture \ref{mainconj} holds, then every linear-invariant subspace-hereditary property is testable with one-sided tester.
\end{corollary}
Observe that if Conjecture \ref{mainconj} holds, then Corollary \ref{cor:alltestable} would give yet
another surprising similarity between linear-invariant properties of boolean functions and graph
properties, since it is known \cite{AS08} that every hereditary graph property is
testable. Actually, as we discuss in the next subsection, if Conjecture 4 holds, then an even
stronger similarity would follow.

Many interesting properties of the hypercube that have been studied for testability are
linear-invariant.  Important examples include linearity \cite{BLR}, being a polynomial of low degree
\cite{AKKLR}, and low Fourier dimensionality and sparsity \cite{GOSSW}.  These properties have all been 
shown to be testable.  Moreover, they all turn out to be subspace-hereditary.  Thus, if our
Conjecture \ref{mainconj} is true, as we strongly believe, then we could explain the testability
of all these properties through a unified perspective that uses no features of these properties
other than their linear invariance.  Note that our main result, Theorem \ref{thm:main}, already shows
(yet again!) that linearity is testable but from a completely different viewpoint than
used in previous analysis.  Furthermore, to show the testability of low degree polynomials
(a.k.a., Reed-Muller codes), we would only need to resolve
Conjecture \ref{mainconj} for a {\em finite}
\footnote{
The characterization of polynomials of degree
$d$ using forbidden induced equations is shown in Appendix \ref{app:reedmuller}. 
}
family of
forbidden induced systems of equations.
\anote{I removed the sentence about dimensionality and sparsity
only being known to have $2$-sided testers, because we observe later
that subspace-heredity implies they also have $1$-sided testers.}

\ignore{
Regarding the properties of Fourier dimensionality
and sparsity, they are currently only known to have two-sided testers \cite{GOSSW}, while Corollary \ref{cor:alltestable}
will potentially yield one-sided testers, resolving an issue raised in \cite{SudanSurvey}.}

\subsection{The proposed characterization of testable linear-invariant properties}\label{subseccharac}

We now turn to discuss our second result, which based on Conjecture \ref{mainconj} gives a
characterization of the linear-invariant properties of Boolean functions that can be tested with
one-sided error using ``natural'' testing algorithms. Let us start with formally defining the types
of ``natural'' testers we consider here. 

\begin{definition}[Oblivious Tester]\label{defoblivious}
An {\em oblivious tester} for a property $\calp=\{\calp_n\}_n$ is a (possibly 2-sided error) non-adaptive, probabilistic algorithm, which, given a distance parameter $\epsilon$,
and oracle access to an input function $f:\F^n_2 \to \{0,1\}$, performs the following steps:
\begin{enumerate}
\item Computes an integer $d = d(\eps)$.  If $d(\eps) > n$, let $H = \F_2^n$.  Otherwise, let $H
  \leq \F_2^n$ be a subspace of dimension  $d(\eps)$ chosen uniformly at random.
\item Queries $f$ on all elements $x \in H$.
\item Accepts or rejects based only on the outcomes of the received answers, the value of $\eps$,
  and its internal randomness.
\end{enumerate}
\end{definition}

We now discuss the motivation for considering the above type of algorithms. The fact that the tester
is non-adaptive and queries a random linear subspace is without loss of generality 
(see Proposition
\ref{prop:canonical}); 
 this is analogous to the fact \cite{AFKS, GoldreichTrevisan} that one can assume a graph property
tester makes its decision only by inspecting a randomly chosen induced subgraph.
The only essential restriction we place on oblivious testers is that their behavior cannot depend on
the value of $n$, the domain size of the input function. If we allow the testing algorithm to make
its decisions based on $n$, then  it can do very strange and unnatural things. For example, we can
now consider properties that depend on the parity of $n$. As was shown in
\cite{ASseparation}, the algorithm can use the size of the input in order to compute the optimal
query complexity. All these abnormalities will not allow us to give any meaningful
characterization. As observed in \cite{AS08} by restricting the algorithm to make its decisions
while not considering the size of the input, we can still test any (natural) property while at the
same time avoid annoying technicalities. We finally note that all the testing algorithms for
testable properties of Boolean functions in prior works were indeed oblivious, and that furthermore
many of them implicitly consider only oblivious testers. In particular, these types of testers were considered in \cite{SudanSurvey}.

As it turns out, oblivious testers can potentially\footnote{The potential relies on the
  validity of Conjecture \ref{mainconj}.} test properties which are slightly more general than
subspace-hereditary properties.
These are defined as follows.

\begin{definition}[Semi Subspace-Hereditary Property]\label{def:semihereditary}
A property $\calp=\{\calp_n\}_n $ is {\em semi subspace-hereditary} if there exists a subspace-hereditary property $\calh$ such that
\begin{enumerate}
\item Any function $f$ satisfying $\calp$ also satisfies $\calh$.
\item There exists a function $M:(0,1)\to \N$ such that for every $\eps \in (0,1)$, if $f:\F^n_2 \to \{0,1\}$ is $\epsilon$-far from
satisfying $\calp$ and $n \geq M(\epsilon)$, then $f|_V$ does not satisfy $\calh$.  \ignore{there
  exists a subspace $V \subseteq \F^n_2$ such that $f|_V$ does not satisfy $\calh$.}
\end{enumerate}
\end{definition}

The intuition behind the above definition is that a semi subspace-hereditary property can only deviate from being ``truly'' subspace-hereditary on functions over a finite domain, where the finiteness is controlled by the
function $M$ in the definition. Our next theorem connects the notion of oblivious testing and semi
subspace-hereditary properties.  Assuming Conjecture \ref{mainconj}, it essentially
characterizes the linear-invariant properties that are testable with one-sided error, thus resolving
Sudan's problem raised in \cite{SudanSurvey}.  

\begin{theorem}\label{thm:charac}
If Conjecture \ref{mainconj} holds, then a linear-invariant property ${\cal P}$ is testable by a
one-sided error oblivious tester {\bf if and only if} ${\cal P}$ is semi subspace-hereditary. 
\end{theorem}

Getting back to the similarity to graph properties, we note that \cite{AS08} obtained a similar
characterization for the graph properties that are testable with one-sided error.  
Let us close by mentioning two points. The first is that most linear-invariant properties are known
to be testable with one-sided error, and hence the question of characterizing these properties is
well motivated.  In fact, for the subclass of linear-invariant properties which also themselves form
a linear subspace, \cite{BHR05} showed that the optimal tester is always one-sided and non-adaptive.
Our second point is that it is natural to ask if there are  linear-invariant properties which are
not testable. A linear-invariant property with query complexity $\Omega(2^n)$ arises
implicitly from the arguments of \cite{GGR}; see Section
\ref{sec:conclusion} for a brief sketch. A second, more natural, example comes from Reed-Muller
codes. \cite{BKSSZ} shows that for any $1 \ll q(n) \ll n$ the linear-invariant property of
being a $\log_2(q(n))$-Reed-Muller code cannot be tested with $o(q(n))$ queries.  We also conjecture
that the property of two functions being isomorphic upto linear transformations of the variables is
not a testable property.  Lower bounds for isomorphism testing have been studied both in the Boolean
function model \cite{FKRSS04, BO10} and in the dense graph model \cite{FischerIsom}, but our problem
specifically does not seem to have been examined in a property testing setting.

\ignore{

Let us now give another motivation for studying
such properties. The characterizations of \cite{AFNS06,BCLSSV06} of the testable graph properties, use the fact that when testing graph properties one can always assume that the testing algorithm behaves in a ``canonical'' way (essentially non-adaptive algorithms). This was observed by \cite{AFKS} and strengthened by Goldreich and Trevisan \cite{GoldreichTrevisan}. As it turns out, when considering linear-invariant properties we can prove a similar result. The following is the appropriate notion of a canonical tester for testing Boolean functions.

\begin{definition}[Canonical Tester]\label{defcanonical}
A {\em canonical tester} is a tester that works by sampling $q=q(\epsilon)$ vectors $x_1,\ldots,x_q$ querying the input function $f$ all vectors in $\mspan(x_1,\ldots,x_q)$ and
answering based on these answers (NEEDS SOME MORE WORK).
\end{definition}

\begin{proposition}\label{propcanonical}
If a linear-invariant property is testable then it is testable by a canonical tester.
\end{proposition}

\begin{definition}[Strongly Canonical Tester]\label{defcanonical}
A {\em strongly canonical tester} is a canonical tester that has the following additional properties:
\begin{itemize}
\item The tester has one-sided error.
\item The algorithms decisions are independent of the size of the input.
\end{itemize}
\end{definition}

}

\subsection{Paper overview}\label{subsecorganize}

The rest of the paper is organized as follows. In Section \ref{sec:uniformity}, we discuss the
regularity lemma of Green \cite{Green05}.  Just as the graph regularity lemma of Szemer\'edi
\cite{Szemeredi} guarantees that every graph can be partitioned into a bounded number of
pseudorandom graphs, Green's regularity lemma guarantees a similar partition for Boolean
functions. This lemma, whose proof relies on Fourier analysis over $\F^n_2$, was used in
\cite{Green05} to show that properties defined by forbidding a single (non-induced) equation are
testable. This basic approach falls short of being able to handle an infinite number of forbidden
non-induced equations or even a single forbidden induced equation.  We thus need to develop a
variant of Green's regularity lemma that is strong enough to
allow such applications. This new variant is described in Section \ref{sec:uniformity}. The
overall approach is motivated by that taken by Alon et al. \cite{AFNS06} in their formulation of the functional
graph regularity lemma. However, the proof here is somewhat more involved since we need to
develop several tools in order to make the approach work. One of them is a certain Ramsey type
result for $\F^n_2$ which is key to our proof and that may be useful in other settings (see
Theorem \ref{thm:ramsey}).  The approach of \cite{AFNS06} only allows one to handle a {\em finite}
number of forbidden subgraphs, which translates in our setting to being able to handle a finite
number of forbidden equations. So, one last technique we employ is motivated by the ideas from
\cite{AS08} on how to handle an infinite number of forbidden subgraphs. This (somewhat complicated)
technique is described in Section \ref{sec:hereditary}. We believe that these set of ideas will prove to
be instrumental in resolving Conjecture \ref{mainconj}. Section \ref{sec:conclusion} is
devoted to some concluding remarks and open problems.

\section{Pseudorandom Partitions of the Hypercube \label{sec:uniformity}}

The {\em support} of a Boolean function $f$ refers to the subset of the domain on which $f$
evaluates to $1$.
If $H$ is a subspace of $\F_2^n$ and given function $f: H\to \zo$, let $\rho(f)$, the {\em density} of $f$,
denote $\frac{\sum_{x \in H} f(x)}{|H|}$.  Recall that the Fourier coefficients of $f$, defined for
each  $\alpha \in H^*$, are:
\begin{equation*}
\widehat{f}(\alpha) = \E_{x \in H} \left[f(x)\cdot (-1)^{\langle x,\alpha\rangle}\right]
\end{equation*}
For a parameter $\eps \in (0,1)$, we say $f$ is {\em $\eps$-uniform} if $\max_{\alpha \neq 0}
|\widehat{f}(\alpha)| < \eps$.  This definition captures the notion of correlation with a  linear
function on $H$, and it will serve as our definition of pseudorandomness.

Given a function $f: \F_2^n \to \zo$, a subspace $H \leq \F_2^n$ and an element $g \in \F_2^n$, define
the function $f_H^{+g} : H \to \zo$ to be $f_H^{+g}(x) = f(x+g)$ for $x \in H$.  The
support of $f_H^{+g}$ represents the intersection of the support of $f$ with the coset $g + H$.  The following lemma
shows that if a uniform function is restricted to a coset of a subspace of low codimension, then the
restriction does not become too non-uniform and its density stays roughly the same.

\begin{lemma}\label{lem:coset-uniformity}
Let $f: \F_2^n \to \zo$ be an $\eps$-uniform function of density $\rho$, and let $H\leq \F_2^n$ be a
subspace of codimension $k$. Then for any $c\in \F_2^n$, the function $f_H^{+c}: H \to \zo$ is
$(2^k \eps)$-uniform and of density $\rho_c$ satisfying $|\rho_c-\rho|< 2^k \epsilon$.
\end{lemma}

\begin{proof}
Let  $H^{\perp}=\{\alpha\in \F_2^n |~\langle\alpha,  h\rangle=0~ \forall h\in H\}$ be the dual to
the vector space $H$, and let $H'=\F_2^n/H$ be the quotient of $H$ in $\F_2^n$.
We wish to show that, for every $c\in H'$, the Fourier coefficients of   $f_H^{+c}$ are small.

For every $\beta\in \F_2^n/H^{\perp}$ and $\alpha \in H^{\perp}$:
\begin{align*}
{\widehat f}(\beta+\alpha) = \E_{x \in \F_2^n} ~\left[f(x) \chi_{\beta+\alpha}(x)   \right]
                         = \E_{c'\in H'} ~ \E_{h\in H}{f_H^{+c'}}(h) \chi_{\beta+\alpha}(c'+h)
                         &= \E_{c'\in H'} \chi_{\beta+\alpha}(c') \E_{h\in H}{f_H^{+c'}}(h) \chi_{\beta}(h) \\
                         &= \frac{1}{2^k} \sum_{c'\in H'} \chi_{\beta+\alpha}(c'){\widehat{f}_H^{+c'}}(\beta)
\end{align*}
Recall that  $\sum_{\alpha\in H^{\perp}} \chi_{\alpha}(c')=\begin{cases} 0, \mbox{ if } c'\not = 0\\
1, \mbox{ if } c'=0
.\end{cases}
$
Fixing  $\beta\in \F_2^n/H^\perp$ and $c\in H'$ and summing up the quantity computed above over all $\alpha\in H^{\perp}$, we obtain
\begin{eqnarray*}
2^k \left( \sum_{\alpha \in {H^{\perp}}} \chi_{\beta+\alpha}(c) {\widehat f}(\beta+\alpha) \right )&=& \sum_{c'\in H'} \sum_{\alpha \in {H^{\perp}}} \chi_{\beta+\alpha}(c+c')   {\widehat f}_H^{+c'}(\beta)\\
&=& \sum_{\alpha \in {H^{\perp}}}  \chi_{\beta+\alpha}(0){\widehat f}_H^{+c}(\beta)+ \sum_{c'\in H'-\{c\}}\sum_{\alpha \in {H^{\perp}}} \chi_{\beta+\alpha}(c+c')   {\widehat f}_H^{+c'}(\beta)\\
&=& 2^k  {\widehat f}_H^{+c}(\beta) + \sum_{c'\in H'-\{0\}} \sum_{\alpha \in {H^{\perp}}} \chi_{\beta+\alpha}(c') {\widehat f}_H^{+c'+c}(\beta)\\
&=&  2^k {\widehat f}_H^{+c}(\beta) + \sum_{c'\in H'-\{0\}} \chi_{\beta}(c') \left(\sum_{\alpha \in {H^{\perp}}} \chi_{\alpha}(c')\right)  {\widehat f}_H^{+c'+c}(\beta)\\
&=&  2^k {\widehat f}_H^{+c}(\beta).
\end{eqnarray*}

Furthermore,
\begin{eqnarray*}
\left|  {\widehat f_H^{+c}}(\beta) \right | = \left| \sum_{\alpha \in {H^{\perp}}} \chi_{\beta+\alpha}(c) {\widehat f}(\beta+\alpha) \right |
                                               \leq  \sum_{\alpha \in {H^{\perp}}} \left|  \chi_{\beta+\alpha}(c) {\widehat f}(\beta+\alpha)  \right|
                                               = \sum_{\alpha \in {H^{\perp}}} \left| {\widehat f}(\beta+\alpha)  \right|
\end{eqnarray*}

Since $f$ is $\epsilon$-uniform, setting $\beta=0$ in the above inequality shows that $|\rho_c - \rho| \leq \sum_{0 \neq
  \alpha \in H^\perp} |\widehat{f}(\alpha)| < 2^k \eps$.  For nonzero $\beta$ in
$\F_2^n/H^\perp$, it follows again from $\eps$-uniformity that $|{\widehat f_H^{+c}}(\beta)|< 2^k \eps$.
\end{proof}

For a subspace $H \leq \F_2^n$, the {\em $H$-based partition} refers to the partitioning of $\F_2^n$ into
the cosets in $\F_2^n/H$.   If $H' \leq H$, then the $H'$-based partition is called a {\em refinement} of the
$H$-based partition.  The {\em order} of the $H$-based partition is defined to be $[G : H]$, i.e., the
index of $H$ as a subgroup or the dimension of the quotient space $\F_2^n/H$.  Using this
notation, Green's regularity lemma can be stated as follows.

\begin{lemma}[Green's Regularity Lemma \cite{Green05}]\label{lem:regularity}
For every $m$ and $\eps > 0$, there exists $T = T_{\ref{lem:regularity}}(m,\eps)$ such that the
following is true.   Given function $f : \F_2^n \to \zo$ with $n > T$ and $H$-based partition of $\F_2^n$
with order at most $m$, there exists a refined $H'$-based partition of order $k$, with $m \leq k
\leq T$, for which $f_{H'}^{+g}$ is not $\eps$-uniform for at most $\eps 2^n$ many $g \in \F_2^n$.
\end{lemma}
\ignore{\begin{remark}
The only difference between Lemma \ref{lem:regularity} and the regularity lemma as formulated
in \cite{Green05} is that here, we require the final partition to be a refinement of a given
partition. It is easy to modify Green's proof to yield our lemma, merely by making the initial
partition of the iterative partitioning procedure be the $H$-based partition instead of the trivial
partition (all of $\F_2^n$).
\end{remark}}

Our main tool in this work is a functional variant of Green's regularity lemma, in which the
uniformity parameter $\eps$ is not a constant but rather an arbitrary function of the order of the partition.
It is quite analogous to a similar lemma, first proved in \cite{AFKS}, in the graph property testing
setting.  The recent work \cite{GT10} shows a (very strong) functional regularity lemma in
the arithmetic setting but it applies over the integers and not $\F_2$.

\begin{lemma}[Functional regularity lemma]\label{lem:functionalreg}
For integer $m$ and function $\mathcal{E} : \Z^+ \to (0,1)$, there exists $T =
T_{\ref{lem:functionalreg}}(m,\mathcal{E})$ such that the following is true.  Given function $f :
\F_2^n \to \zo$ with $n \geq T$, there exist subspaces $H' \leq H \leq \F_2^n$ that satisfy:
\begin{itemize}
\item
Order of $H$-based partition is $k \geq m$, and order of $H'$-based partition is $\ell \leq T$.
\item
There are at most $\mathcal{E}(0)\cdot 2^n$ many $g \in \F_2^n$ such that $f_H^{+g}$ is not
$\mathcal{E}(0)$-uniform.
\item
For every $g \in \F_2^n$, there are at most $\mathcal{E}(k)\cdot 2^{n-k}$ many $h \in H$ such that
$f_{H'}^{+g+h}$ is not $\mathcal{E}(k)$-uniform.
\item
There are at most $\mathcal{E}(0)\cdot 2^n$ many $g \in \F_2^n$ for which there are more than
$\mathcal{E}(0)\cdot 2^{n-k}$ many $h \in H$ such that $|\rho(f_H^{+g}) - \rho(f_{H'}^{+g+h})| >
\mathcal{E}(0)$.
\end{itemize}
\end{lemma}

\begin{proof}
Let us first give an informal overview of the proof.  The basic idea is to repeatedly apply Lemma
\ref{lem:regularity}, at each step refining the partition obtained in the previous step.  At each
step, Lemma \ref{lem:regularity} is applied with a uniformity parameter that depends on the order of
the partition obtained in the previous step.  We stop
when the {\em index} of the partitions stop increasing substantially.  Given a subspace $H$, the
index of the $H$-based partition is defined to be the variance of the densities in the cosets:
\begin{equation*}
\ind(f,H)~\eqdef~ \frac{1}{2^n} \sum_{g \in \F_2^n} \rho^2(f_H^{+g})
\end{equation*}
We show that when the indexes of two successive partitions are close, then on average, each coset of
the finer partitioning has roughly the same density as the coset of the coarser partitioning it is
contained in.

To implement the above ideas, we need the following two claims about the index of partitions.  Their
proofs are essentially identical to those for the corresponding Lemmas 3.6 and 3.7
respectively in \cite{AFKS}, and so we are a bit brief in the following.
\begin{claim}\label{clm:defect}
Given subspace $H \leq \F_2^n$ and function $f : \F_2^n \to \zo$, suppose that there are at least $\eps
2^n$ many $g \in \F_2^n$ such that $|\rho(f) - \rho(f_H^{+g})| > \eps$.  Then:
\begin{equation*}
\ind(f, H) > \rho^2(f) + \frac{\eps^3}{2}
\end{equation*}
\end{claim}
\begin{proof}
Observe that the average of $\rho(f_H^{+g})$ over all $g \in
\F_2^n$ equals $\rho(f)$.  From our assumptions, either there are $\frac{\eps}{2}2^n$ many $g \in
\F_2^n$ such that $\rho(f) - \rho(f_H^{+g}) > \eps$ or there are  $\frac{\eps}{2}2^n$ many
$g \in \F_2^n$ such that $\rho(f) - \rho(f_H^{+g}) < -\eps$.  For either case, we
can use the defect form of the Cauchy-Schwarz inequality to prove our claim.
\end{proof}

\begin{claim}\label{clm:indextodensity}
For function $f : \F_2^n \to \zo$ and subspaces $H' \leq H \leq \F_2^n$, suppose the $H$-based
partition of order $k$ and its refinement, the $H'$-based partition, of order $\ell$ satisfy
$\ind(f,H') - \ind(f,H) \leq \frac{\eps^4}{2}$ for some $\eps$.  Then, there are at most $\eps 2^n$
many $g \in \F_2^n$ for which there are more than $\eps 2^{n-k}$ many $h \in H$ satisfying
$|\rho(f_H^{+g}) - \rho(f_{H'}^{+g+h})| > \eps$.
\end{claim}
\begin{proof}
Suppose that there are $> \eps 2^n$ many $g \in \F_2^n$ such that there are $> \eps 2^{n-k}$ many $h
\in H$ satisfying $|\rho(f_H^{+g}) - \rho(f_{H'}^{+g+h})| > \eps$.  Use Claim \ref{clm:defect} to
obtain a contradiction:
\begin{align*}
\ind(f,H')
= \frac{1}{2^\ell} \sum_{u \in \F_2^n/H'} \rho^2(f_{H'}^{+u})
&= \frac{1}{2^k} \sum_{v \in \F_2^n/H} \frac{1}{2^{\ell-k}} \sum_{h \in H/H'}\rho^2(f_{H'}^{+v+h})\\
&= \frac{1}{2^k} \sum_{v \in \F_2^n/H} \ind(f_H^{+v})\\
&> \frac{1}{2^k} \left(\sum_{v \in \F_2^n/H} \rho^2(f_H^{+v}) + \eps\cdot 2^k
  \frac{\eps^3}{2}\right) \\
&= \ind(f,H) + \frac{\eps^4}{2}
\end{align*}
\end{proof}

Now we have the pieces needed to prove the lemma.  We can assume $\mathcal{E}(\cdot)$ is monotone
non-increasing.  Let $\eps = \mathcal{E}(0)$.  We define $T$ inductively as follows.  Let $T^{(1)} =
T_{\ref{lem:regularity}}(m,\eps)$, and for $i > 1$, let:
\begin{equation*}
T^{(i)} = T_{\ref{lem:regularity}}\left(T^{(i-1)},
  \mathcal{E}\left(T^{(i-1)}\right)\cdot 2^{-T^{(i-1)}} \right)
\end{equation*}
Set $T = T_{\ref{lem:functionalreg}}(m,\mathcal{E}) \eqdef T^{(2\eps^{-4} + 1)}$.

We now show that this choice of $T$ suffices.  Given function $f : \F_2^n \to \zo$, apply Lemma
\ref{lem:regularity} with $m$ and $\eps$ to get a subspace $H_1$, and thereafter repeatedly apply it
to get a sequence of finer subspaces $H_2, H_3, H_4, \dots$, with $H_1 \geq H_2 \geq H_3 \geq H_4
\geq \cdots$, by invoking Lemma \ref{lem:regularity} at each step $i > 1$ with $T^{(i-1)}$ and
$\mathcal{E}\left(T^{(i-1)}\right)\cdot 2^{-T^{(i-1)}}$ as the two input parameters.  Stop when
$\ind(f, H_{i+1})  - \ind(f, H_i) < \frac{\eps^4}{2}$.  This happens when $i$ is at most $2\eps^{-4} + 1$ because the
index of any partition is less than $1$.  Let $H = H_i$ and $H' = H_{i+1}$.  It's clear that the
codimension $k$ of $H$ at least $m$ and that the codimension $\ell$ of $H'$  is at most $T$.  The
second item in the lemma follows from the uniformity guarantee of Lemma \ref{lem:regularity}
and from the fact that $\mathcal{E}(T^{(i-1)}) < \mathcal{E}(0)$.  For the third, note that Lemma
\ref{lem:regularity} guarantees that there are at most $\mathcal{E}(k) 2^{-k} 2^n = \mathcal{E}(k)
2^{n-k}$ values of $g \in \F_2^n$ such that $f_{H'}^{+g}$ is not $(\mathcal{E}(k) 2^{-k})$-uniform
and, hence, not $\mathcal{E}(k)$-uniform.  So, clearly, there are at most so many $g$ contained in
any coset of $H$.   Finally, the fourth item follows from Claim \ref{clm:indextodensity}. This completes the proof
of Lemma \ref{lem:functionalreg}.
\end{proof}

We use Lemma \ref{lem:functionalreg} in two main ways.  For one of them, we use the lemma directly.
For the other, we use the following simple but extremely useful corollary which allows us to say
that there are many cosets in a partitioning which, on the one hand, are {\em all} uniform, and on
the other hand, are arranged in an algebraically nice structure.

\begin{corollary}\label{cor:allreg}
For every $m$ and $\mathcal{E} : \Z^+ \to (0,1)$, there exist $T =
T_{\ref{cor:allreg}}(m,\mathcal{E})$ and $\delta = \delta_{\ref{cor:allreg}}(m, \mathcal{E})$ such
that the following is true.  Given function $f : \F_2^n \to \zo$ with $n \geq T$, there exist
subspaces $H' \leq H \leq \F_2^n$ and an injective linear map $I : \F_2^n/H \to \F_2^n/H' $
such that:
\begin{itemize}
\item
The $H$-based partition is of order $k$, where $m \leq k \leq T$.  Additionally, $|H'| \geq \delta
2^n$.

\item
For each $u \in \F_2^n/H$, $I(u) + H'$ lies inside the coset $u + H$.  Note that $I(0) = 0$ since
$I$ is linear.

\item
For every nonzero $u \in \F_2^n/H$, the set $f_{H'}^{+I(u)}$ is $\mathcal{E}(k)$-uniform.

\item
There are at most $\mathcal{E}(0) 2^n$ many $g \in \F_2^n$ for which $|\rho(f_H^{+g}) -
\rho(f_{H'}^{+I(u)})| > \mathcal{E}(0)$ where $u = g \pmod H$.
\end{itemize}
\end{corollary}

\begin{proof}
We can assume $\mathcal{E}$ is a nonincreasing function.  Denote $\mathcal{E}(0)$ as $\eps$, and set
$\mathcal{E}'(r) = \min(\mathcal{E}(r),\frac{\eps}{6}, \frac{1}{2^{r+1}})$.  We will show that $T =
T_{\ref{cor:allreg}}(m,\mathcal{E}) \eqdef T_{\ref{lem:functionalreg}}(m,\mathcal{E}')$ and $\delta
= \delta_{\ref{cor:allreg}}(m,\mathcal{E}) \eqdef 1/2^T$ suffice for our proof.

Apply Theorem \ref{lem:functionalreg} with $m$ and the function $\mathcal{E}'$ as inputs.  Let $H$ and $H'$ be the
subspaces obtained there, for the given $f : \F_2^n \to \zo$.  We find $I$ satisfying the
conditions of the claim exists using the probabilistic method.

Fix $k$ linearly independent elements $u_1,\dots,u_k \in \F_2^n/H$ (viewing $\F_2^n/H$ as a vector
space over $\F_2$).  For every $i \in
[k]$, choose independently and uniformly at random an element $v$  from $H/H'$ and let $I(u_i)$
equal $u_i + v + H'$.  The value of $I$ over the rest of $\F_2^n/H$ is
determined by linearity, as the $u_i$'s form a basis for $\F_2^n/H$.  It's immediate that $I(u) +
H'$ lies inside $u+ H$ for every $u \in \F_2^n/H$.

Observe that unless $u = 0$, each $I(u)+H'$ is uniformly distributed among the cosets of $H'$ lying in $u+H$.  Hence, for any nonzero
$u$, the probability that $f_{H'}^{+I(u)}$ is not $\mathcal{E}(k)$-uniform is at most $1/2^{k+1}$, by our
choice of parameters.  Applying the union bound, the probability that there exists nonzero $u \in
\F_2^n/H$ such that $f_{H'}^{+I(u)}$ is not $\mathcal{E}(k)$-uniform is at most $1/2$.  Also, the expected number
of $g \in \F_2^n$, with $u = g \pmod H$, for which $|\rho(f_H^{+g}) - \rho(f_{H'}^{+I(u)})| > \eps$
is at most $\frac{\eps}{6} 2^n + \frac{\eps}{6}2^n + 1 \leq \frac{\eps}{2}2^n$, and hence by the
Markov inequality, with probability at least $\frac{1}{2}$, the number of $g \in \F_2^n$ satisfying
this condition is at most $\eps 2^n$. Therefore, there must exist a choice of $I$ making both the
third and fourth claims true.
\end{proof}

The next lemma is in a similar spirit to Corollary \ref{cor:allreg}.  It also obtains a set of
uniform cosets which are structured algebraically, but in this case, all of them are contained
inside the same subspace. 

\begin{lemma}\label{lem:allreginside}
For every positive integer $d$ and $\gamma \in (0,1)$, there exists $\delta =
\delta_{\ref{lem:allreginside}}(d,\gamma)$ such that the following is true.  Given $f: \F_2^n \to
\zo$, there exists a subspace $H \leq \F_2^n$ and a subspace $K$ of dimension $d$ in the quotient
space $\F_2^n/H$ with the following properties:
\begin{itemize}
\item
$|H| \geq \delta 2^n$.
\item
For every nonzero $u \in K$, $f_H^{+u}$ is $\gamma$-uniform.
\item
Either $\rho(f_H^{+u}) \geq \frac{1}{2}$ for every nonzero $u \in K$ or $\rho(f_H^{+u}) < \frac{1}{2}$
for every nonzero $u \in K$.
\end{itemize}
\end{lemma}

 We need a different set of tools to prove this lemma.  Specifically, we
use linear algebraic variants of the classic theorems of Tur\'an and Ramsey. We note that the (classic) 
Tur\'an and Ramsey Theorems are key tools in many applications of the graph regularity lemma, for example
in the well known bound on the Ramsey numbers of bounded degree graphs \cite{CRST}. Hence, the variants that we use of these
classic results may be useful in other applications of Greens's regularity lemma. 


\ignore{\subsection{Helpful Tur\'an and Ramsey-type results}}
\begin{proposition}[Tur\'an theorem for subspaces]\label{thm:turan}
For positive integers $n$, if $S$ is a subset of $\F_2^n$ with density greater than $1-
\frac{1}{2^{d-1}}$, then there exists a subspace $H \leq \F_2^n$ of dimension $d$ such that $H - \{0\}$ is
contained in $S$.  Moreover, there is a subset of $\F_2^n$ with density $\left(1-\frac{1}{2^{d-1}}
\right)$ which  does not contain $H - \{0\}$ for any subspace $H \leq \F_2^n$.
\end{proposition}

\begin{proof}
Let $S \subseteq \F_2^n$ be a maximal set that does not contain $H - \{0\}$ for any $d$-dimensional
subspace $H$.  Since $S$ is maximal, it must contain $K-\{0\}$ for some $(d-1)$-dimensional subspace
$K$ (if not, we can simply add it to $S$ without introducing points of $H-\{0\}$ for any
$d$-dimensional subspace $H$).  Let $K'$ be an $(n-d+1)$-dimension subspace that
intersects $K$ only at $\{0\}$.

Now, observe that for any nonzero $\alpha \in K'$, at least one of the elements
of $\{\alpha + k : k \in K\}$ must not belong to $S$.  Otherwise, $S$ would contain $(K-\{0\}) \cup
\{\alpha + k : k \in K\} = H - \{0\}$ for a $d$-dimensional subpace $H = \text{span}(K \cup
\{\alpha\})$, contradicting our assumption for $S$.  Thus, we can upper-bound the number of points
in $S$ by:
\begin{equation*}
|S| \leq |K' - \{0\}|\cdot (|K| - 1) + |K - \{0\}| = (2^{n-d+1}-1)\cdot (2^{d-1} - 1) + (2^{d-1} -1 ) = 2^n - 2^{n-d+1}
\end{equation*}

To see that the above bound is tight, let $S = \F_2^n - K'$ for any $(d-1)$-dimensional
subspace $K \leq \F_2^n$ and $K'$ as above.  It is easy to check that this $S$ does not
contain $H-\{0\}$ for any $H\leq \F_2^n$ with $\dim(H) = d$.
\end{proof}

\begin{theorem}[Ramsey theorem for subspaces]\footnote{As pointed to us recently by Noga Alon, this theorem might be implied by the Folkman-Rado-Sanders Theorem, but we include a self-contained proof for the sake of completeness.}\label{thm:ramsey}
For every positive integer $d$, there exists $N = N_{\ref{thm:ramsey}}(d)$ such that for any subset $S \subseteq
\F_2^N$, there exists a subspace $H \leq \F_2^N$ of dimension $d$ such that  $H - \{0\}$ is
contained either in $S$ or in $\bar{S}$.
\end{theorem}

\begin{proof}
We will show a stronger statement, which we describe in the following lemma.

\begin{lemma}\label{lem:double-dim}
For every positive integer $d_1, d_2$, there exists $N(d_1, d_2)$ such that for any subset $S \subseteq
\F_2^{N(d_1, d_2)}$, either there exists a subspace $H_1 \leq \F_2^{N(d_1, d_2)}$ of dimension $d_1$ such that  $H_1-\{0\}$ is
contained in $S$ or there exists a subspace $H_2 \leq \F_2^{N(d_1, d_2)}$ of dimension $d_2$  such that $H_2-\{0\}$ is contained in
 $\bar{S}$.
\end{lemma}

One can immediately deduce the statement of the theorem by taking $d=d_1=d_2$ in
Lemma~\ref{lem:double-dim}. To prove Lemma~\ref{lem:double-dim} we first prove the following helpful
result. For a subspace $H\leq \F_2^n$ we say that an affine subspace $a+H$ is {\em strict} if $a\in
\F_2^n/H-\{0\}$.

\begin{lemma}\label{lem:affine-ramsey}
For every positive integer $d$, there exists $N_a = N_a(d)$ such that for any subset $S \subseteq
\F_2^{N_a}$, there exists a strict affine subspace $A \leq \F_2^{N_a}$ of dimension $d$ such that  $A$ is
contained either in $S$ or in $\bar{S}$.
\end{lemma}

\begin{proof}
Notice that $N_a(1)=1$.  Assume, by induction that the lemma holds for dimension $d-1$, and let
$N_a(d)={2^{N_a(d-1)+1}}+N_a(d-1)$.  Let $S\subseteq \F_2^{N_a(d)}$ be an arbitrary set, let
$H=\F_2^{N_a(d-1)}$, and $H'=\F_2^{N_a(d)}/H$. Notice that $|H'|=2^{2^{N_a(d-1)}+1}$.  For each $c\in
H' -\{0\}$ consider the set $f_H^{+c}\subset H$. Since there are $2^{2^{N_a(d-1)+1}}-1$ possible
such sets, and each set has size at most $2^{N_a(d-1)}$ it follows that there exists $c_1 \not=
c_2\in H'-\{0\}$ such that $f_H^{+c_1}=f_H^{+c_2}$.  By the induction hypothesis, either
$f_H^{+c_1}$ or its complement contains a $d-1$ dimensional affine subspace.  Assume w.l.o.g. that
$f_H^{+c_1}$ contains an affine subspace $\alpha+f_{d-1}$ of dimension $d-1$ (otherwise replace $S$
by ${\bar{S}}$), for some $\alpha \in H-f_{d-1}$. Then the affine subspaces $\alpha+c_1+f_{d-1}$ and
$\alpha+c_2+f_{d-1}$ are both contained in $S$. Let $A_d=(\alpha+c_1+f_{d-1})\cup
(\alpha+c_2+f_{d-1}) \subset S$. To conclude the proof, notice that $A_d=\alpha+c_1+\mspan( c_2-c_1, f_{d-1})$ is a strict
affine subspace of dimension $d$, since $\alpha\not =c_1$ and $c_2-c_1\not\in f_{d-1}$.
\end{proof}

\bigskip

\begin{proofof}{Lemma~\ref{lem:double-dim}}
The proof follows by induction on $d_1$ and $d_2$, with the base cases $N(0,1)=N(1,0=1$.
Assume that there exists $N(d_1-1, d_2)$ and $N(d_1, d_2-1)$ satisfying the conditions of the lemma. Define
$$N(d_1, d_2)=N_a(\max (N(d_1-1, d_2), N(d_1, d_2-1))),$$ where $N_a(d)$ is the quantity defined in Lemma~\ref{lem:affine-ramsey}.
We show that for any arbitrary set $S\subseteq \F_2^{N(d_1, d_2)}$ either it contains a subspace of dimension $d_1$ (except $0$) or its complement contains a subspace of dimension $d_2$ (except $0$). Suppose $N(d_1-1, d_2)\geq N(d_1, d_2-1)$.
By definition and by Lemma~\ref{lem:affine-ramsey}, there exists a strict affine subspace $A\subseteq \F_2^{N(d_1, d_2)}$ such that $A=a+H\subseteq S$ or $A\subseteq {\bar S}$ (where $H$ is the subspace underlining $A$).
Assume for now that the former holds. Since $H\cap S\subseteq \F_2^{N(d_1-1, d_2)}$,  by the induction hypothesis, either $H\cap S$ contains a subspace of dimension $d_1-1$ or $H-S$ contains a subspace of dimension $d_2$, in which case we are done. If $H\cap S$ contains a subspace $f_{d_1-1}-\{0\}$ of dimension $d_1-1$, then define $f_{d_1}= f_{d_1-1} \cup a+ f_{d_1-1}=\mspan(a, f_{d_1-1})$. Clearly $f_{d_1}\in S$ and it has dimension $d_1$, which completes the proof of this case.
It remains to deal with the case when $A\subseteq {\bar S}$. Since $N(d_1-1, d_2)\geq N(d_1, d_2-1)$, there exists another affine subspace $A'=a'+H'\subset A\subseteq {\bar S}$ of dimension $N(d_1, d_2-1)$. Again, by the induction hypothesis, the set $H'\cap  S$ either contains a subspace of dimension $d_1$, in which case we are done, or $H'-{S}$ contains a subspace $f_{d_2-1}$ of dimension $d_2-1$. In the latter case define $f_{d_2}=f_{d_2-1}\cup a'+f_{d_2-1}=\mspan(a', f_{d_2-1})$. Finally, notice that $f_{d_2}\in {\bar S}$ and it has dimension $d_2$.
\end{proofof}


This concludes the proof of Theorem~\ref{thm:ramsey}.\end{proof}

Given these results, Lemma~\ref{lem:allreginside} follows fairly readily. 

\bigskip
\begin{proofof}{Lemma \ref{lem:allreginside}}
Set $\delta = \delta_{\ref{lem:allreginside}}(d,\gamma) \eqdef 2^{- T_{\ref{lem:regularity}}(r,\min(2^{-r-2},\gamma))}$ with
$r = N_{\ref{thm:ramsey}}(d)$.  Given $f : \F_2^n \to \zo$, apply Lemma \ref{lem:regularity} with
inputs $r$ and $\min(2^{-r-2},\gamma)$ to obtain a subspace $H$ such
that restrictions of $S$ to at most $2^{-r-2}$ fraction of the cosets of the $H$-based partition are
not $\gamma$-uniform. Using Proposition \ref{thm:turan}, there exists a subspace $L \leq \F_2^n/H$
of dimension $r$ such that for every nonzero $u \in L$, the set $f_{H}^{+u}$ is $\gamma$-uniform.
Furthermore, since $L$ is of dimension $N_{\ref{thm:ramsey}}(d)$, by Theorem \ref{thm:ramsey}, there
exists a subspace $K \leq L \leq \F_2^n/H$ satisfying the final condition of the lemma.
\end{proofof}

\section{Forbidding Infinitely Many Induced Equations\label{sec:hereditary}}
In this section, we prove our main result (Theorem \ref{thm:main}) that properties characterized by
infinitely many forbidden induced equations are testable.
To begin, let us fix some notation.  Given a matrix $M$ over $\F_2$ of
size $m$-by-$k$, a string $\sigma \in \zo^k$, and a function $f : \F_2^n \to \zo$, if there exists
$x = (x_1,\dots,x_k) \in (\F_2^n)^k$ such that $Mx = 0$ and $f(x_i) = \sigma_i$ for all $i \in [k]$,
we say that $f$ {\em induces $(M,\sigma)$ at $x$} and denote this by $(M,\sigma) \mapsto f$.

  The following theorem is the core of the proof of Theorem \ref{thm:main}.

\begin{theorem}\label{thm:modif}
For every infinite family of equations $\calf = \{(E^1,\sigma^1),(E^2,\sigma^2),\dots,(E^i,\sigma^i),\dots\}$
  with each $E^i$ being a row vector $[1 ~1~\cdots~1]$ of size $k_i$ and $\sigma^i \in \zo^{k_i}$
  a $k_i$-tuple, there are functions $N_\calf(\cdot)$, $k_\calf(\cdot)$ and
  $\delta_\calf(\cdot)$ such  that the following is true for any $\eps \in (0,1)$.  If a function
  $f : \F_2^n \to \zo$ with $n > N_\calf(\eps)$ is $\eps$-far from being $\calf$-free, then $f$
  induces $\delta \cdot 2^{n(k_i-1)}$ many copies of some $(E^i,\sigma^i)$, where $k_i \leq k_\calf(\eps)$
  and $\delta \geq \delta_\calf(\eps)$.
\end{theorem}

Armed with Theorem \ref{thm:modif}, our main theorem becomes a straightforward consequence.  We
postpone the proof of this, because we will prove a stronger fact in Section \ref{subsec:extension}.
\ignore{ 
\begin{proofof}{Theorem~\ref{thm:main}}
 Theorem \ref{thm:modif} allows us to devise the
following tester $T$ for $\calf$-freeness.
$T$, given input $f : \F_2^n \to \zo$, first checks if $n \leq N_\calf(\eps)$, and in this case, it
queries $f$ on the entire domain and
decides accordingly.  Otherwise, $T$ repeats the following test $O(1/\delta_\calf(\eps))$ many
times: for every $i$ such that $k_i \leq k_\calf(\eps)$, independently and uniformly at random
choose elements $x_1,\dots,x_{k_i-1} \in \F_2^n$, set $x_{k_i} = x_1 + x_2 + \cdots x_{k_i-1}$ and
reject immediately if $f(x_j) = \sigma^{i}_j$ for every $j \in [k_i]$.  $T$ accepts if it never
rejects in any of the iterations.  It's clear that the query complexity of $T$ is constant and that
$T$ always accepts if the input is $\calf$-free.  It rejects inputs $\eps$-far from $\calf$-free
because Theorem \ref{thm:modif} guarantees that there will be an equation of size at most
$k_\calf(\eps)$ for which $T$ will detect solutions to, with constant probability.
\end{proofof}
} 
To start the proof of Theorem \ref{thm:modif}, let us relate pseudorandomness (uniformity) of a
function to the number of solutions to a single equation induced by it. Similar and more general
statements have been shown previously, but we need only the following claim for what follows.
\begin{lemma}[Counting Lemma]\label{lem:count}
For every $\eta \in (0,1)$ and integer $k>2$, there exist $\gamma = \gamma_{\ref{lem:count}}(\eta,k)$ and
$\delta = \delta_{\ref{lem:count}}(\eta,k)$ such that the
following is true.  Suppose $E$ is the row vector $[1~ 1 \cdots 1]$ of size $k$, $\sigma \in \zo^k$
is a tuple, $H$ is a subspace of $\F_2^n$, and $f:\F_2^n \to \zo$ is a function. Furthermore, suppose there are $k$
not necessarily distinct elements $u_1,\dots,u_k \in \F_2^n/H$ such that $Mu = 0$ where $u =
(u_1,\dots,u_k)$, $f_H^{+u_i}: H \to \zo$ is $\gamma$-uniform for all $i \in [k]$, and $\rho(f_H^{+u_i})$ is at
least $\eta$ if $\sigma(i) = 1$ and at most $1-\eta$ if $\sigma(i) = 0$ for all $i \in [k]$.   Then,
there are at least $\delta |H|^{k-1}$ many $k$-tuples $x = (x_1, x_2, \dots, x_k)$, with each $x_i
\in u_i + H$, such that $f$ induces $(E,\sigma)$ at $x$.
\end{lemma}

\ignore{\begin{remark}
In fact, the above proof approach can be used to count the number of times $(M,\sigma)$ is induced
by $f$, for any matrix $M$ that represents a {\em linear system of complexity $1$}
(as defined by Green and Tao in \cite{GT06}).  See Section \ref{sec:conclusion} for more details.
\end{remark}}

\begin{proof}
Fix $v_1 \in u_1 + H$, $v_2 \in u_2 + H, \dots, v_k \in u_k+H$ such that $v_1 + v_2 + \cdots + v_k =
0$; there exist such $v_i$'s because $u_1 + u_2 + \cdots + u_k = 0$ in the quotient space $\F_2^n/H$.
Define Boolean functions $f_1,\dots,f_k : H \to \zo$ so that $f_i(x) = {f_H^{+v_i}}(x)$ if
$\sigma(i) = 1$ and $f_i(x) = 1 - f_H^{+v_i}(x)$ if $\sigma(i) = 0$.  By our assumptions,
$\widehat{f_i}(0) \geq \eta$ and each $|\widehat{f_i}(\alpha)| < \gamma$ for all $\alpha \neq 0$.
Now, observe that, using $\gamma$-uniformity and Cauchy-Schwarz, we have:
\begin{align*}
\E_{x_1,\dots,x_{k-1} \in H} &\left[f_1(x_1)f_2(x_2)\cdots f_{k-1}(x_{k-1}) f_k(x_1 + x_2 + \cdots +
  x_{k-1})\right]\\
&= \sum_{\alpha \in H^*} \widehat{f_1}(\alpha) \widehat{f_2}(\alpha)\cdots \widehat{f_k}(\alpha)\\
&\geq \eta^k - \sum_{\alpha \neq 0} |\widehat{f_1}(\alpha) \widehat{f_2}(\alpha)\cdots \widehat{f_k}(\alpha)| \\
&\geq \eta^k - \gamma^{k-2} \sqrt{\sum_{\alpha}|\widehat{f_1}(\alpha)|^2} \sqrt{\sum_{\alpha}|\widehat{f_2}(\alpha)|^2} \\
&\geq \eta^k - \gamma^{k-2}
\end{align*}
Setting $\gamma = \gamma_{\ref{lem:count}}(\eta,k)\eqdef (\eta^k/2)^{1/(k-2)}$ makes the
above expectation at least $\eta^k/2$.  Now note that every $x_1,\dots,x_k \in H$ such that $x_1
+ \cdots + x_k = 0$ gives $y = (y_1, \dots, y_k)$, where $y_i = v_i + x_i$ for all $i \in [k]$, such
that $f$ induces $(E,\sigma)$ at $y$.  Thus, we have from above that there are at least $\delta
|H|^{k-1}$ many such $y$'s, where $\delta = \delta_{\ref{lem:count}}(\eta,k) \eqdef \eta^k/2$.
\end{proof}

\subsection{Proof of Theorem~\ref{thm:modif}}
\label{sec:modifproof}

Before seeing the full technical details of the proof of Theorem~\ref{thm:modif} we proceed with a more intuitive overview.

In light of Lemma~\ref{lem:count}, our strategy  will be to partition the domain
into uniform cosets, using Green's regularity lemma (Lemma \ref{lem:regularity}) in some fashion,
and then to use the above counting lemma to count the number of induced solutions to some equation in
$\calf$.  But one issue that immediately arises is that, because $\calf$ is an infinite family of
equations, we do not know the size of the equation we would want the input function to induce.
Since  Lemma \ref{lem:count} needs different uniformity parameters to count equations of different
lengths, it is not {\em a priori} clear how to set the uniformity parameter in applying the
regularity lemma.  (If $\calf$ was finite, one could set the uniformity parameter to correspond to
the size of the largest equation in $\calf$.) 

To handle the infinite case, our basic approach will be to classify the input function into one of a
finite set of classes.  For each such class $c$, there will be an associated number $k_c$ such that
it is guaranteed that any function classified as $c$ must induce an equation in $\calf$ of size at most
$k_c$.  If there is such a classification scheme, then we know that {\em any} input function must
induce an equation of size at most $\max_c k_c$.  How do we perform this classification?  We
use the regularity lemma.  Consider the following idealized situation. Fix an integer $r$.  Suppose
we could modify the input $f: \F_2^n\to \zo$ at a small fraction of the domain to get a function $F:
\F_2^n \to \zo$ and then could apply Lemma \ref{lem:regularity} to get a partition of
order $r$ so that the restrictions of $F$ to each coset was exactly $0$-uniform.  $F$ is then a
constant function (either $0$ or $1$) on each of the $2^r$ cosets, and so, we can classify
$F$ by a Boolean function $\mu : \F_2^r \to \{0,1\}$ where $\mu(x)$ is the value of $F$ on the coset
corresponding to $x$.  Notice that there are only finitely many such $\mu$'s.  Since $F$ differs
from $f$ at only a small fraction of the domain and since 
$f$ is far from $\calf$-free, $F$ must also induce some equation in $\calf$.  Then, for every such
$\mu$ and corresponding $F$,  there is a smallest equation in $\calf$ that is induced by
$F$.  We can let $\Psi_\calf(r)$ be the maximum
over all such $\mu$ of the size of the smallest equation in $\calf$ that is induced by the $F$
corresponding to $\mu$.  We then might hope that this function $\Psi_\calf(\cdot)$ can be used to tune the uniformity
parameter by using the functional variant of the regularity lemma (Lemma \ref{lem:functionalreg}).

There are a couple of caveats.  First, we will not be able to get the restrictions to
every coset to look perfectly uniform.  Second, if $F$ induces solutions to an equation, it does
not necessarily follow that $f$ also does.  To get around the first problem, we use the fact that
Lemma \ref{lem:count} is not very restrictive on the density conditions.  We think of
the uniform cosets which have density neither too close to $0$ nor $1$ as ``wildcard'' cosets at
which both the restriction of $f$ and its complement behave pseudorandomly and have
non-negligible density. Thus, the $\mu$ in the above paragraph will map into $\{0,1,*\}^r$, where a
`$*$' denotes a wildcard coset.  For the second  problem, note that it is not really a problem if
$\calf$-freeness is known to be monotone.  In this case, $F$ inducing an equation automatically
means $f$ also induces an equation, if we obtained $F$ by removing elements from the support of
$f$.  For induced freeness properties, though, this is not the case.  Using ideas from
\cite{AFKS} and the tools from Section \ref{sec:uniformity}, we structure the modifications from $f$
to $F$ in such a way so as to force $f$ to 
induce solutions of an equation if $F$ induces a solution to the same equation. 
We elaborate
much more on this issue during the course of the proof.

The observations described in the proof sketch above motivate the following definitions. 
\begin{definition}\label{def:partialinduce}
Given function $\mu: \F_2^r \to \{0,1,*\}$, a $m$-by-$k$ matrix $M$ and a $k$-tuple $\sigma \in
\zo^k$, suppose there exist $x_1,\dots,x_k \in \F_2^r$ such that
$Mx = 0$ where $x=(x_1,\dots,x_k)$, and for every $i \in [k]$, $\mu(x_i)$ equals either $\sigma(i)$
or $*$.  In this case, we say $\mu$ {\em partially induces $(M,\sigma)$ at $x$} and denote this by
$(M,\sigma) \mapsto_* \mu$.
\end{definition}

\begin{definition}\label{def:psi}
Given a positive integer $r$ and an infinite family of systems of equations $\calf =
\{(M^1,\sigma^1),(M^2,\sigma^2),\dots\}$ with $M^i$ being a $m_i$-by-$k_i$ matrix of rank $m_i$ and
$\sigma^i \in \zo^{k_i}$ a $k_i$-tuple, define $\calf_r$ to be the set of
functions $\mu: \F_2^r \to \{0,1,*\}$ such that there exists some $(M^i,\sigma^i) \in \calf$ with
$(M^i,\sigma^i) \mapsto_* \mu$.  Given $\calf$  and integer $r$ for which $\calf_r \neq
\emptyset$, define the following function:

\begin{equation*}
  \Psi_\calf(r) ~\eqdef ~ \max_{\mu \in \calf_r} \min_{\{(M^i,\sigma^i) : (M^i,\sigma^i)
    \mapsto_* \mu\}} k_i
\end{equation*}
\end{definition}

\bigskip

\begin{proofof}{ Theorem \ref{thm:modif}}
Define the function $\mathcal{E}$ by setting $\mathcal{E}(0) = \eps/8$ and for any $r > 0$:
\begin{equation*}
\mathcal{E}(r) = \delta_{\ref{lem:allreginside}}(\Psi_\calf(r),\gamma_{\ref{lem:count}}(\eps/8,
\Psi_\calf(r))) \cdot \min(\eps/8,
\gamma_{\ref{lem:count}}(\eps/8, \Psi_\calf(r)))
\end{equation*}
Additionally, let $T(\eps) = T_{\ref{cor:allreg}}(8/\eps,\mathcal{E})$, and set $N_\calf(\eps)
\eqdef T(\eps)$.  Also, set $k_\calf(\eps) \eqdef \Psi_\calf(T(\eps))$ and
\begin{equation*}
\delta_\calf(\eps) \eqdef\left(\delta_{\ref{lem:allreginside}}(\Psi_\calf(r),
  \gamma_{\ref{lem:count}}(\eps/8, \Psi_\calf(r))) \cdot \delta_{\ref{cor:allreg}}(8/\eps,
  \mathcal{E}) \right)^{\Psi_\calf(\eps)} \cdot \delta_{\ref{lem:count}}(\eps/8,\Psi_\calf(T(\eps)))
\end{equation*}
We proceed to show that these parameter settings suffice.

Suppose we are given input function $f : \F_2^n \to \zo$ with $n > N_\calf(\eps) = T_{\ref{cor:allreg}}(8/\eps,
\mathcal{E})$.  As mentioned in the paragraphs preceding the proof, our strategy will be to
partition the domain in such a way that we can find cosets in the partition satisfying the
conditions of Lemma \ref{lem:count}.  To this end, we apply Corollary \ref{cor:allreg} with $8/\eps$
and the function $\mathcal{E}$ as inputs.  This yields subspaces $H' \leq H \leq \F_2^n$ and linear map $I
: \F_2^n/H \to \F_2^n/H'$, where the order of the $H$-based partition, which we denote $\ell$, satisfies
$8/\eps \leq \ell \leq  T_{\ref{cor:allreg}}(8/\eps,\mathcal{E})$.  Recall that $I(u) + H'$ is
contained in $u + H$ for every coset $u \in \F_2^n/H$.
Observe that from our setting of parameters, we have that for every {\em nonzero} $u \in \F_2^n/H$, the
restriction $f_{H'}^{+I(u)}$ is
$(\delta_{\ref{lem:allreginside}}(\Psi_\calf(\ell),\gamma_{\ref{lem:count}}(\eps/8,
\Psi_\calf(\ell))) \cdot \gamma_{\ref{lem:count}}(\eps/8, \Psi_\calf(\ell)))$-uniform.

But we have no such uniformity guarantee for $f_{H'}^{+0}$.  This would not pose an obstacle if
$\calf$-freeness were a monotone property (i.e., if each $\sigma^i$ equalled $1^{k_i}$).  If that
were the case, we could simply make $f$ zero on all elements of $H$.  Since $H$ is still only a small
fraction of the domain, the modified function would still be far from $\calf$-free, and we would be
guaranteed that remaining solutions to equations of $\calf$ induced by $f$ would only use elements from
cosets of $H$ for which we have a guarantee about the corresponding coset of $H'$.  But if
$\calf$-freeness is not monotone, such a scheme would not work, since it's not clear at all how to
change the value of $f$ on $H$ so that any solution to an equation from $\calf$ would only involve
elements from nonzero shifts of $H$.

To resolve this issue, we further partition $H'$ to find affine subspaces within $H'$ on which we
can guarantee that the restriction of $f$ is uniform.  The idea is that once we know that there is a
solution involving $H$, we are going to look not at $H'$ itself but at the smaller affine subspace
within $H'$ on which $f$ is known to be uniform.  Specifically, apply Lemma \ref{lem:allreginside}
to $f_{H'}^{+0}$ with input parameters $\Psi_\calf(\ell)$ and  $\gamma_{\ref{lem:count}}(\eps/8,
\Psi_\calf(\ell))$.  This yields subspaces $H''$ and $W$, both of which contained in $H'$, such that
$|H''| \geq \delta_{\ref{lem:allreginside}}(\Psi_\calf(\ell),\gamma_{\ref{lem:count}}(\eps/8,
\Psi_\calf(\ell))) |H'|$ and $\dim(W/H'') = \Psi_\calf(\ell)$.  We further know
that for every nonzero $v \in W/H''$, the function $f_{H''}^{+v}$ is $\gamma_{\ref{lem:count}}(\eps/8,
\Psi_\calf(\ell))$-uniform.

Now, let's ``copy'' $W$ on  cosets $I(u) + H'$ for every $u \in
\F_2^n/H$.  We do this by specifying\footnote{One way to accomplish this is to define $J$
  appropriately for $\ell$ linearly independent elements of $\F_2^n/H$ and then use linearity to
  define it on all of $\F_2^n/H$.  } another linear
map $J : \F_2^n/H \to \F_2^n$ so that for any $u \in \F_2^n/H$, the coset\footnote{Note that the
  image of $J$ is to elements of $\F_2^n$ and not $\F_2^n/W$, even
  though  we think of the output as denoting a coset of $W$.  The reason is that we will find it
  convenient to fix the shift and not make it modulo $W$.} $J(u) + W$ lies inside
$I(u) + H'$ (which itself lies inside $u + H$).  Each coset $J(u)+W$ also has an $H''$-based
partition of order $\Psi_\calf(\ell)$, just as $W$ itself does.  Consider $v \in \F_2^n/H''$ such
that $v + H''$ lies inside $J(u) + W$ for some nonzero $u \in \F_2^n/H$.  Then, because we know the
uniformity of $f_{H'}^{+I(u)}$ and we have a lower bound on the size of $H''$, it follows from Lemma
\ref{lem:coset-uniformity}  that $f_{H''}^{+v}$ is $\gamma_{\ref{lem:count}}(\eps/8,
\Psi_\calf(\ell))$-uniform.  Thus, for any nonzero $v \in \F_2^n/H''$ such that $v + H''$ lies
inside $J(u) + W$ for some $u \in \F_2^n/H$, it is the case that $f_{H''}^{+v}$ is $\gamma_{\ref{lem:count}}(\eps/8,
\Psi_\calf(\ell))$-uniform.

In the following, we will show how to apply Lemma \ref{lem:count} on some of these cosets $f_{H''}^{+v}$.  We have
already argued their uniformity above.  We now need to make sure that the pattern of their densities
allow Lemma \ref{lem:count} to infer many induced copies of some equation in $\calf$.  To this end,
we modify $f$ to construct a new function $F : \F_2^n \to \zo$.  $F$ is initially identical to $f$
on the entire domain, but is then modified in the following order:
\begin{enumerate}
\item
For every nonzero $u \in \F_2^n/H$ such that $|\rho(F_{H}^{+u}) - \rho(F_{H'}^{+I(u)})| > \eps/8$,
do the following.  If $\rho(F_{H'} ^{+I(u)}) \geq \frac{1}{2}$, then make $F(x) = 1$ on all $x \in
u+H$.  Otherwise, make $F(x)= 0$ on all  $x \in u+H$.

\item
For every nonzero $u \in \F_2^n/H$ such that $\rho(F_{H'}^{+I(u)}) > 1-\eps/4$, make $F(x)=1$ for
all  $x \in u+H$.  On the other hand, if $u\in \F_2^n/H$ is nonzero and $\rho(F_{H'}^{+I(u)}) <
\eps/4$, make $F(x) = 0$ for all  $x \in u+H$.

\item
If for all nonzero $v \in W/H''$, $\rho(F_{H''}^{+v}) \geq \frac{1}{2}$, then make $F(x) = 1$ for
all  $x \in H$.  On the other hand, if for all nonzero $v \in W/H''$, $\rho(F_{H''}^{+v})
< \frac{1}{2}$, them make $F(x) = 0$ for all  $x \in H$.  (One of these two
conditions is true by construction.)
\end{enumerate}

The following observation shows that $F$ also must induce solutions to some equation from $\calf$,
since $F$ is $\eps$-far from being $\calf$-free.
\begin{claim}
$F$ is $\eps$-close to $f$.
\end{claim}
\begin{proof}
We count the number of elements added or removed at each step of the modification.  For the first
step, Corollary \ref{cor:allreg} guarantees that at most $\mathcal{E}(0) \leq \eps/8$ fraction of
cosets $u+H$ have $|\rho(F_{H}^{+u}) - \rho(F_{H'}^{+I(u)})| > \eps/8$.  So, $F$ is modified in at
most $\frac{\eps}{8} 2^n$ locations in the first step.  In the second step, if
$1 > \rho(F_{H'}^{+I(u)}) > 1-\eps/4$, then $\rho(F_H^{+u}) > 1-3\eps/8$ because the first step has
been completed.  Similarly, if $0 < \rho(F_{H'}^{+I(u)}) < \eps/4$, then $\rho(F_H^{+u}) <
3\eps/8$.  So, $F$ is modified in at most $\frac{3\eps}{4}2^n$ locations in the second step.  As
for the third step, $H$ contains at most $2^{n-\ell} \leq 2^{n - 8/\eps} < \frac{\eps}{8} 2^n$
elements for $\eps \in (0,1)$.  So, in all, $F$ is $\eps$-close to $f$.
\end{proof}

Now, we define a function $\mu: \F_2^\ell \to \{0,1,*\}$ based on $F$ and argue that it must
partially induce solutions to some equation in $\calf$.  Since $H$ is of codimension $\ell$,
$\F_2^n/H \cong \F_2^\ell$ and we identify the two spaces.  For $u \in \F_2^n/H$, if $F(x) = 1$ on
the entire coset $u+H$, let $\mu(u) = 1$.  On the other hand, if $F(x) = 0$ on the entire coset
$u+H$, then let $\mu(u) = 0$.  In any other case, let $\mu(u) = *$.

\begin{claim}
There exists $(E^i,\sigma^i) \in \calf$ such that $(E^i,\sigma^i) \mapsto_* \mu$.
\end{claim}
\begin{proof}
As already observed, $F$ is not $\calf$-free, and let $(E^i,\sigma^i)\in \calf$ be some equation whose
solution is induced by $F$ at $(x_1,\dots,x_{k_i}) \in (\F_2^n)^{k_i}$.  Now let $y =
(y_1,\dots,y_{k_i}) \in (\F_2^\ell)^{k_i}$ where for each $j \in [k_i]$, $y_j = x_j \pmod H$.  It's
clear that $E^i y = 0$.  To argue that $F$ partially induces $\mu$ at $y$, suppose for
contradiction that for some $j  \in [k_i]$, $\mu(y_j) = 0$ but $\sigma^i_j = 1$.  But if $\mu(y_j) =
0$, then $F$ is the constant function $0$ on all of $y_j + H$, contradicting the existence of $x_j
\in y_j + H$ with $F(x) = 1$.  We get a similar contradiction if $\mu(y_j) = 1$ but $\sigma^i_j = 0$.
\end{proof}

Using Definition \ref{def:psi}, we immediately get that there is some $(E^i,\sigma^i) \in \calf$ of
size at most $\Psi_\calf(\ell)$ such that $(E^i,\sigma^i) \mapsto_* \mu$.  Fix $x_1,\dots,x_{k_i}
\in \F_2^n$ where $F$ induces $(E^i,\sigma^i)$, and as in the above proof, let $y_1,\dots,y_{k_i}
\in \F_2^n/H$ where each $y_j = x_j \pmod H$.  Also, pick $k_i-1$ linearly independent elements
$\tilde{v}_1,\dots, \tilde{v}_{k_i-1}$ from $W/H''$, which is possible since $\dim(W/H'') =
\Psi_\calf(\ell) > k_i - 1$, and choose $v_1 \in \tilde{v}_1 + H'', \dots, v_{k_i-1} \in
\tilde{v}_{k_i-1} + H''$ such that $v_1,\dots,v_{k_i}$ are linearly independent.  Additionally set
$v_{k_i} = \sum_{j=1}^{k_i-1} v_j$.  Notice that none of $v_1,\dots,v_{k_i}$
are in $H''$.  Now, consider the sets $f_{H''}^{+J(y_1) + v_1}, f_{H''}^{+J(y_2) + v_2}, \dots,
f_{H''}^{+J(y_{k_i}) + v_{k_i}}$. (Notice these are restrictions of $f$, not $F$!)   We will
show that these sets respect the density and uniformity conditions for Lemma \ref{lem:count} to
apply.

As for uniformity, we have already argued that each of these sets is
$\gamma_{\ref{lem:count}}(\eps/8, \Psi_\calf(\ell))$-uniform, since $J(y_j) + v_j$ is not in $H''$ for
every $j \in [k_i]$.  For density, we argue as follows.
For every $j \in [k_i]$, there are three cases: $\mu(y_j) = 1$, $\mu(y_j) = 0$, and $\mu(y_j) = *$.
Consider the first case.  If $y_j + H$ was affected by the first modification from $f$ to $F$,
then, $\rho(f_{H'}^{+I(y_j)}) \geq
\frac{1}{2}$, and using the $\mathcal{E}(\ell)$-uniformity of $f_{H'}^{+I(y_j)}$ along with Lemma
\ref{lem:coset-uniformity}, we get that $\rho(f_{H''}^{+J(y_j) + v_j}) \geq \frac{1}{2} -
\mathcal{E}(\ell) \cdot \delta^{-1}_{\ref{lem:allreginside}}(\Psi_\calf(r),\gamma_{\ref{lem:count}}(\eps/8,
\Psi_\calf(r))) \geq \frac{1}{2} - \frac{\eps}{8} \geq \frac{\eps}{8}$.  If $y_j + H$ was affected by the second
modification, then, by the same argument, we get that $\rho(f_{H''}^{+J(y_j) +
  v_j}) \geq 1 - \frac{\eps}{4} - \frac{\eps}{8} \geq \frac{\eps}{8}$.  Else, if $y_j + H$ was
affected by the third modification from $S$ to $S'$, we are automatically guaranteed that  $\rho(f_{H''}^{+J(y_j) +
  v_j}) \geq \frac{1}{2}$ since $J(y_j) + v_j \not\in H''$.  The case $\mu(y_j) = 0$ is similar, and the analysis
shows that $\rho(f_{H''}^{+J(y_j) + v_j}) \geq 1-\frac{\eps}{8}$.  Finally, consider the
``wildcard'' case, $\mu(y_j) = *$.  This case arises only if $y_j \neq 0$ and $\eps/4 \leq
\rho(f_{H'}^{+I(y_j)}) \leq 1-\eps/4$.  Again using $\mathcal{E}(\ell)$-uniformity of $f_{H'}^{+I(y_j)}$ along with Lemma
\ref{lem:coset-uniformity}, we get that $\eps/8 \leq \rho(f_{H''}^{+J(y_j) + v_j}) \leq 1-\eps/8$.

Thus, we can apply Lemma \ref{lem:count} with $\eps/8$ and $\Psi_\calf(\ell)$ as the parameters to
get that there are at least $\delta_{\ref{lem:count}}(\eps/8, \Psi_\calf(\ell)) |H''|^{k_i-1}$
tuples $z = (z_1,\dots,z_{k_i})$ with each $z_j \in J(y_j) + v_j + H''$ at which $(E^i,\sigma^i)$ is
induced .  Finally, each such $z_1,\dots,z_{k_i}$ leads to a distinct $z' = (z_1',\dots,z_{k_i}')
\in (\F_2^n)^{k_i}$ at which $(E^i,\sigma^i)$ is induced by $f$, by setting each $z_j'$ to
$J(y_j) + v_j + z_j$ and observing that $\sum_{j=1}^{k_i} J(y_j) + v_j =  J\left(\sum_{j=1}^{k_i} y_j\right) +
\sum_{j=1}^{k_i} v_j = 0$. This completes the proof of Theorem \ref{thm:modif}.
\end{proofof}

\subsection{Extending to Systems of Equations of Complexity $1$}\label{subsec:extension}

As mentioned in the introduction, the result we actually prove is stronger than Theorem
\ref{thm:main}.  To describe the full set of properties for which we can show testability, we first
need to make the following definition.
\begin{definition}[Complexity of linear system \cite{GT06}]
An $m \times k$ matrix $M$ over $\F_2$ is said to be of {\em
 (Cauchy-Schwarz) complexity $c$}, if $c$ is the smallest positive integer
for which the following is true.  For every $i \in [k]$, there exists
a partition of $[k]\backslash \{i\}$ into $c+1$ subsets $S_1,\cdots,S_{c+1}$
such that for every $j \in [c+1]$, $\left(\mathbf{e}_i + \sum_{i' \in S_j}
\mathbf{e}_{i'}\right) \not \in \mathsf{rowspace}(M)$, where
$\mathsf{rowspace}(M)$ is the linear subspace of $\F_2^k$ spanned by
the rows of $M$.  
\end{definition}
In other words, if we view the rowspace of the matrix $M$ as
specifying a collection of linear dependencies on $k$ variables $x_1,\dots,x_k$, then $M$ has
complexity $c$ if for every variable $x_i$, the rest of the variables
$x_1,\dots,x_{i-1},x_{i+1},\dots,x_k$ can be partitioned into $c+1$ sets $S_1,\dots,S_{c+1}$ such
that $x_i$ is not linearly dependent on the variables of any single $S_j$. Let us make a few
remarks to illustrate the definition.  Green and Tao show (Lemma 1.6 in \cite{GT06}) that if
each of these linear dependencies involves more than two variables, then the complexity of
$M$ is at most $\mathsf{rank}(M) = m$.  In particular then, if $M$ has one row and is nonzero on more
than two coordinates, $M$ has
complexity $1$.  This is the setting we discussed in the introduction.  We slightly extend this
observation in the claim below.  Before we state
it, we observe that in the context of property testing, 
it is only natural to exclude matrices which yield linear dependencies involving less than three
variables. If the rowspace of the matrix $M$ contains a vector which is nonzero at only one 
coordinate $i$, then for any string $\sigma$ of length $k$, the property of $(M,\sigma)$-freeness
must contain all functions $f$ such that $f(0) = 
1-\sigma_i$, and so {\em every} function is exponentially close to such a property.  Similarly, if
$\mathsf{rowspace}(M)$ contains a vector nonzero only at two coordinates $i$ and $j$, then for any
$\sigma \in \{0,1\}^k$, either $(M,\sigma)$-freeness is trivial (if $\sigma_i \neq \sigma_j$) or it is
equivalent to $(M',\sigma')$-freeness where $\sigma'$ is the string obtained by removing coordinate
$j$ and $M'$ is the matrix obtained by removing column $j$,
adding $1$ (mod $2$) to every element in column $i$ and row-reducing the resulting matrix.

\begin{claim}\label{lem:2rows}
If $M \in \F_2^{m \times k}$ is a matrix with two rows such that every vector in its
rowspace has at least three nonzero coordinates, then $M$ has complexity $1$.
\end{claim}
\begin{proof}
Let $R_1 \subseteq [k]$ be the set of coordinates for which the first row is nonzero, and $R_2
\subseteq [k]$ those for which the second row is nonzero.  We can assume that $R_1 \not \subseteq
R_2$ and $R_2 \not \subseteq R_1$, because if, say, $R_1 \subseteq R_2$, we could replace the second
row by the sum of the first and second, making $R_1$ and $R_2$ disjoint but preserving the rowspace
of the matrix.  Also, we we can assume w.l.o.g. that $R_1 \cup R_2 = [k]$.

Fix $i \in [k]$.  We want to show a partition of $[k]\backslash \{i\}$ into sets $S_1$, $S_2$ such
that $\mathbf{e}_i + \sum_{i' \in S_1} \mathbf{e}_{i'} \notin \mathsf{rowspace}(M)$ and similarly
for $S_2$.  If $i \in R_1\backslash R_2$, let $S_1$ consist of two elements, one from
$R_2\backslash R_1$ and one from $R_1\backslash\{i\}$, and let $S_2$ be the rest.
If $i \in R_2\backslash R_1$, let $S_1$ consist of one element from  $R_1\backslash R_2$ and one
from $R_2 \backslash \{i\}$, and let $S_2$ be the rest.   And finally, if $i \in R_1 \cap R_2$, let
$S_1$ consist of one element from $R_1 \backslash R_2$ and one from $R_2 \backslash R_1$, and let
$S_2$ be the rest.  It is straightforward to check that the definition of complexity $1$ is
satisfied by these choices.
\end{proof}

\noindent
More generally, an infinitely large class of complexity $1$ linear systems is generated by {\em
  graphic matroids}. We refer the reader to \cite{BCSX09} for definition and details.  That this
class contains the class of matrices proved to be of complexity $1$ in Claim \ref{lem:2rows}
is easy to show.  We proved the claim separately above only to be self-contained without introducing
matroid notation.    One final remark is that if $M$ is the matrix in the characterization of
Reed-Muller codes of order $d$ from Appendix \ref{app:reedmuller}, then $M$ has complexity exactly
$d$; see Example $3$ of \cite{GT06}.


Our main result in this section is the extension of  Theorem \ref{thm:main}  to  complexity 1
systems of equations. 

\begin{theorem}\label{thm:mainstrong}
Let $\mathcal{F} = \{(M^1,\sigma^1),(M^2,\sigma^2),\dots\}$ be a possibly infinite set
of induced systems of equations, with each $M^i$ of complexity $1$. Then, the property of being ${\cal
  F}$-free is testable  with one-sided error.
\end{theorem}

We next describe how to modify the previous proof to the new setting.
The following analog to Theorem \ref{thm:modif} is the core of the proof of Theorem \ref{thm:mainstrong}.

\begin{theorem}\label{thm:genmodif}
For every infinite family $\calf = \{(M^1,\sigma^1),(M^2,\sigma^2), \dots, (M^i,
\sigma^i), \dots\}$, where each $M^i$ is a $m_i \times k_i$ matrix over $\F_2$ of complexity $1$,
there are functions $N_\calf(\cdot)$, $k_\calf(\cdot)$ and $\delta_\calf(\cdot)$ such 
that the following is true for any $\eps \in (0,1)$.  If a function $f : \F_2^n \to \zo$ with $n >
N_\calf(\eps)$ is $\eps$-far from being $\calf$-free, then $f$ induces $\delta \cdot 2^{n(k_i-m_i)}$
many copies of some $(M^i,\sigma^i)$, where $k_i \leq k_\calf(\eps)$ and $\delta \geq
\delta_\calf(\eps)$.  
\end{theorem}

We show how to deduce Theorem \ref{thm:mainstrong} from Theorem \ref{thm:genmodif} next.  Note that,
as promised earlier, this is also a proof of Theorem \ref{thm:main} assuming Theorem
\ref{thm:modif}. 

\bigskip

\begin{proofof}{Theorem~\ref{thm:mainstrong}}
 \indent Theorem \ref{thm:genmodif} allows us to devise the
following tester $T$ for $\calf$-freeness.
$T$, given input $f : \F_2^n \to \zo$, first checks if $n \leq N_\calf(\eps)$, and in this case, it
queries $f$ on the entire domain and
decides accordingly.  Otherwise, $T$ selects independently and uniformly at random a set $D$ of $d$
elements from $\F_2^n$, where we will specify $d$ at the end of the argument.  It then 
queries all points in the linear subspace spanned by the elements of $D$ and then accepts or rejects
based on whether $f$ restricted to this subspace is $\calf$-free or not.

Clearly, if $f$ is $\calf$-free, then the tester always accepts because the property is
subspace-hereditary.  Also, 
if $n \leq N_\calf(\eps)$, then the correctness of the algorithm is trivial.  So, suppose $f$ is
$\eps$-far from $\calf$-free and $n > N_\calf(\eps)$.  For the $M^i$ guaranteed to exist from
Theorem \ref{thm:genmodif}, let $K$ be a $k_i \times c$ matrix over $\F_2$,
where $c = k_i -m_i \leq k_\calf(\eps)$, such that the columns of $K$ form a basis for
the kernel of $M^i$.  Then, every $y = (y_1,\dots,y_c) \in (\F_2^n)^c$ yields a distinct vector $x =
(x_1,\dots,x_k) \in (\F_2^n)^k$ formed by letting $x = Ky$ that satisfies $M^ix= M^iKy = 0$.
Therefore, because of Theorem \ref{thm:genmodif}, the probability that uniformly chosen
$y_1,\cdots,y_c \in \F_2^n$ yield $x = (x_1,\dots,x_k)$ such that $f$ induces $(M^i,\sigma^i)$ at
$x$ is at least $\delta_\calf(\eps)$.  The probability that $D$ does not contain such
$y_1,\dots,y_c$ is at most $(1-\delta)^{d/c} < e^{\delta_\calf(\eps) d/c}<1/3$ if we choose $d =
O(c/\delta_\calf(\eps)) = O(k_\calf(\eps)/\delta_\calf(\eps))$.   Thus with probability at least
$2/3$, $\mathsf{span}(D)$ contains $x_1,\dots,x_k$ such that $f$ induces $(M^i,\sigma^i)$ at
$x=(x_1, \dots, x_k)$, making the tester reject.
\end{proofof}

To prove Theorem \ref{thm:genmodif}, the
main ingredient that changes is the counting lemma.

\begin{lemma}[Counting Lemma for Complexity $1$]\label{lem:gencount}
For every $\eta \in (0,1)$ and integer $k>2$, there exist $\gamma =
\gamma_{\ref{lem:count}}(\eta,k)$ and $\delta = \delta_{\ref{lem:count}}(\eta,k)$ such that the
following is true.  Suppose $M$ is an $m \times k$ matrix of complexity $1$ and rank $m < k$, $\sigma \in \zo^k$
is a tuple, $H$ is a subspace of $\F_2^n$, and $f:\F_2^n \to \zo$ is a function. Furthermore, suppose there are $k$
not necessarily distinct elements $u_1,\dots,u_k \in \F_2^n/H$ such that $Mu = 0$ where $u =
(u_1,\dots,u_k)$, $f_H^{+u_i}: H \to \zo$ is $\gamma$-uniform for all $i \in [k]$, and $\rho(f_H^{+u_i})$ is at
least $\eta$ if $\sigma(i) = 1$ and at most $1-\eta$ if $\sigma(i) = 0$ for all $i \in [k]$.   Then,
there are at least $\delta |H|^{k-m}$ many $k$-tuples $x = (x_1, x_2, \dots, x_k)$, with each $x_i
\in u_i + H$, such that $f$ induces $(M,\sigma)$ at $x$.
\end{lemma}

Lemma \ref{lem:gencount} is a special case of the Generalized von Neumann Theorem
(Proposition 7.1 in \cite{GT06}). The rest of the proof is a straightforward modification of
Section \ref{sec:modifproof}.  Namely, whenever the old proof requires $k$ elements or $k$ cosets
$x_1, \dots, x_k$ to satisfy the equation $x_1 + \cdots + x_k = 0$, the new proof would require that
they satisfy the equation $Mx = 0$ where $x = (x_1, \dots, x_k)$.


\section{Characterization of natural one-sided testable properties}\label{app:charac}

We now turn to showing Theorem \ref{thm:charac} which states that for linear-invariant properties,
testability with a one-sided error oblivious tester is equivalent 
to the property being semi subspace-hereditary (recall here Definition \ref{def:semihereditary}).

First we formalize the discussion from the introduction regarding the fact that it is always possible to assume that the testing algorithm for a
one-sided testable linear-invariant property  makes its decision only by querying the
input function on a random linear subspace of constant dimension.

\begin{proposition}\label{prop:canonical}
Let $\calp$ be a linear invariant property, and let $T$ be an arbitrary one-sided tester for $\calp$
with query complexity $d(\eps,n)$.  Then, there exists a one-sided tester $T'$ for $\calp$ that
selects a random subspace $H$ of dimension $d(\eps,n)$, queries the input on all points of $H$, and
decides based on the oracle answers, the value of $\eps$ and $n$, and internal randomness\footnote{Note here, we leave open
  the possibility that the decision of the tester may not be
  based only on properties of the selected subspace.  This gap can be resolved using the same techniques as used by \cite{GoldreichTrevisan}
  for the graph case, but this point is not relevant for our purposes and so we do not elaborate more
here.}.  Note that $T'$ is non-adaptive and has query complexity $2^{d(\eps,n)}$.
\end{proposition}

\begin{proof}
Consider a tester $T_2$ that acts as follows.  If the tester $T$ on the input makes queries
$x_1,\dots,x_d$, then $T'$ queries all points in $\text{span}(x_1,\dots,x_d)$ but makes
its decision based on $x_1,\dots,x_d$ just as $T$ does.  Clearly, $T_2$ is also a one-sided tester
for $\calp$ and with query complexity at most $2^{d(\eps)}$.

Now, define a tester $T'$ as follows.  Given oracle access to a function $f: \F_2^n \to \zo$, $T'$
first selects uniformly at random a non-singular linear transformation $L : \F_2^n \to \F_2^n$, and
then invokes $T_2$ providing it with oracle access to the function $f \circ L$.  That is, when $T_2$
makes query $x$, then algorithm $T'$ makes query $L(x)$.  We argue that the sequence of queries made
by $T'$ are the elements of a uniformly chosen random subspace of dimension at most $d(\eps)$.  To
see this, fix the input $f$ and the randomness of $T_2$.  Then, for each $i \in [2^{d(\eps)}]$ for which
the $i$'th query, $x_i$, made by $T_2$ is linearly independent of the previous $i-1$ queries,
$x_1,\dots,x_{i-1}$, it's the case that $L(x_i)$ is a uniformly chosen random element from outside
$\text{span}(L(x_1),\dots,L(x_{i-1}))$.  So, for every fixing of the random coins of $T_2$, the queries
made by $T'$ span a uniformly chosen subspace of dimension at most $d(\eps)$, and hence, this is
also the case when the coins are not fixed.  $T'$ is a one-sided tester for $\calp$ because if $f \in
\calp$, then $f \circ L \in \calp$ by linear invariance, and if $f$ is $\eps$-far from $\calf$, then
$f \circ L$ is also $\eps$-far from $\calp$ because $L$ is a permutation on $\F_2^n$.
\end{proof}

An oblivious tester, as defined in Definition \ref{defoblivious}, differs from the tester $T'$ of
the above proposition in that the dimension of the selected subspace and the decision made by the
tester are not allowed to depend on $n$.  As argued there, it is very reasonable to expect natural
linear-invariant properties to have such testers, and indeed, prior works have already implicitly
restricted themselves in this way.

We can now proceed with the proof of Theorem \ref{thm:charac}.

\bigskip

\begin{proofof}{Theorem \ref{thm:charac}}
Let us first prove the forward direction of the theorem.  Note that for this direction, we do
not need to assume the truth of Conjecture \ref{mainconj}.
Given a linear-invariant property $\mathcal{P}$
that can be tested with one-sided error by an oblivious tester,
we will build a subspace-hereditary property $\calh$ containing $\calp$, by identifying a (possibly infinite)
collection of matrices $M^i$ and binary strings $\sigma^i$ such that
$\calh$ is equivalent to the property of being $\{(M^i, \sigma^i)\}_i$- free.

Let $\cals$ consist of the pairs $(H,S)$, where $H$ is a subspace of $\F_2^n$ and $S \subseteq H$ is
a subset, that satisfy the following two properties:
(1) $\dim(H)=d(\epsilon)$ for some $\epsilon$, and
(2) if for this $\epsilon$, the tester rejects its input with some positive probability when the
evaluation of its input on the sampled subspace is $\indic_S$.
For $(H,S)\in \cals$ let $d=\dim(H)$.
Consider the matrix $A_H$ over $\F_2$ with each row representing an element of $H$ in some fixed
basis. Notice that $A_H$ is a $(2^{\ell}\times \ell)$-sized matrix. Define $M_H$, a matrix over
$\F_2$  of size
$(2^{\ell}-\ell)\times 2^{\ell}$, such that $M_H A_H=0$. Finally, for each $i\in [2^\ell]$ define
$\sigma_S(i)=\indic_S(x_i)$, where $x_i$ is the element represented in the $i$'th row of $A_H$. Let
$\calm$ be the set of pairs $(M_H, \sigma_S)$ obtained in this way from every $(H,S)\in \cals$.

We now proceed to verify that $\calh$ satisfies the conditions of
Definition~\ref{def:semihereditary}.  To show that $\calp$ is $\calm$-free, let $f\in \calp_n$, and
suppose that there exists $(M_H ,\sigma_S)\in \calm$ such that $(M_H, \sigma_S)\mapsto f$, for some
$\epsilon$, and for some $H$ with $\dim(H)= d(\epsilon)$ and $S \subseteq H$. We show that $f$
is rejected with some positive probability, a contradiction to the fact that the test is
one-sided. If $(M_H,\sigma_S)$ is induced by $f$ at $(x_1,\dots,x_{2^{d(\eps)}})$, then these elements
  necessarily span a $d(\eps)$-dimensional subspace so that the function restricted to that subspace is
  $\indic_S \circ L$ for some linear transformation $L : \F_2^{n} \to \F_2^{d(\eps)}$ (determined by
  the choice of basis that was used to represent $H$).  Thus, this
  immediately implies by the definition of $(M_H,\sigma_S)$ that the tester   rejects $f$ with
  positive probability.

To verify the second part of the Definition~\ref{def:semihereditary}, let $M(\epsilon)=d(\epsilon)$.
Suppose $f:\F_2^n \to \zo$, with $n > M(\eps)$ is $\epsilon$-far from satisfying $\calp$. In this
case, in order for the tester to reject $f$ with positive probability, it must select a
$d(\eps)$-dimensional subspace $H$ so that the restriction to $H$ equals the indicator function on
$S$ (upto a linear transformation), for some $(H,S) \in {\cal S}$.  Therefore $T$ is not
$\calm$-free, and thus $T\not \in \calh$.

It remains to show the opposite direction of Theorem \ref{thm:charac}.  We here assume Conjecture
\ref{mainconj} that every subspace-hereditary property $\calp$ is testable by a one-sided tester.
Our first observation that, in this case, it is actually testable by an {\em oblivious} one-sided
tester.  Namely, we show that the clearly oblivious tester, which checks whether the input function
restricted to a random linear subspace satisfies $\calp$ or not, is a valid tester.  We need to argue that if a
non-oblivious tester rejects input $f$ that is $\eps$-far from $\calp$ by querying its
values on a random $d(\eps)$-dimensional subspace (we already know the 
tester is of this type from Proposition \ref{prop:canonical}), then with high probability, the input function restricted
to a random  $3d(\eps)$-dimensional subspace does not satisfy the property $\calp$.  Suppose it did. 
But then, if the original tester first uniformly selected a $3d(\eps)$-dimensional subspace $H$ and
then uniformly selected a $d(\eps)$-dimension subspace $H'$ inside it, and ran its decision
based on $f|_{H'}$, it will accept the input with large probability, which is a contradiction to the
soundness of the tester since $H''$ is a uniformly distributed $d(\eps)$-dimensional subspace.
Thus, for a testable subspace-hereditary property, we can assume that the tester simply checks for
$\calp$ on the sampled subspace, and is hence, oblivious to the value of $n$.  This argument is
analogous to one of Alon for graph properties, reported in \cite{GoldreichTrevisan}.

Now, assuming that every subspace-hereditary property is testable by an oblivious one-sided  tester
(Conjecture \ref{mainconj}), we wish to show that every semi 
subspace-hereditary property is testable by an oblivious one-sided tester. Let $\calp$ be a a semi
subspace-hereditary property and let $\calh$ be the subspace-hereditary property associated to $\calp$ in
Definition~\ref{def:semihereditary}. By our assumption, $\calh$ has a one-sided tester $T'$, which
on input $\eps$ makes $Q'(\eps)$ queries and rejects inputs $\eps$-far from $\calh$ with probability
$2/3$.
The tester $T$ for $\calp$ makes $Q(\eps)=\max(Q'(\eps/2), 2^{M(\eps/2)})$  queries (where
$M(\cdot)$ comes from Definition~\ref{def:semihereditary}) and proceeds as follows.  If the size of
the input is at most $Q(\eps)$, then by definition, $T$ receives the evaluation of the function all
of the input and in this case, it simply checks if the input belongs to $\calp$.  Otherwise $T$
emulates $T'$ with distance parameter $\eps/2$ and accepts if and only if $T'$ accepts.

Notice that $T$ is one-sided. Indeed,  if the input $f$ satisfies $\calp$  then  $f\in \calh$  and thus $T'$ always accepts, causing $T$ to always accept.
To prove  soundness, we first argue that if $f$ is $\eps$-far from $\calp$ then it is $\eps/2$-far from $\calh$.
Suppose otherwise, and modify $f$ in at most an $\eps/2$ fraction of the domain in order to obtain a function $g\in \calh$. Thus $g$ is still $\eps/2$-far from $\calp$, and by Definition~\ref{def:semihereditary} $g\not \in \calh$, a contradiction.
Finally, since $f$ is $\eps/2$-far from $\calh$ and since $T'$ mistakenly accepts such inputs with probability at most $1/3$ so does $T'$.
\end{proofof}

\section{Concluding Remarks and Open Problems}\label{sec:conclusion}

Obviously, the main open problem we would like to see resolved is Conjecture
\ref{mainconj}. One appealing way to prove the conjecture would be to proceed as we have but to
obtain a stronger notion of pseudorandomness in the regularity lemma.  The notion of
$\eps$-uniformity obtained from Green's regularity lemma corresponds to the Gowers $U^2$ norm,
whereas in order to be able to prove Conjecture \ref{mainconj} in its full generality, we would
presumably need a similar regularity lemma with respect to the Gowers $U^k$ norm \cite{Gow01} for any fixed $k$.
Such a higher order regularity lemma has been very recently obtained by Green and Tao \cite{GT10} over
the integers and over fields of large characteristic.  However, it is not yet available over $\F_2$,
as the inverse conjectures for the Gowers norms over $\F_2$ have not yet been completely clarified \cite{GreenPers}.

Let us mention some other observations and open problems related to this work.

\begin{itemize}
\ignore{
\item As we have mentioned in Subsection \ref{subseccharac}, we can show that any testing algorithm
of a linear-invariant property can be converted into a non-adaptive one. The argument is similar to
the one used by Goldreich and Trevisan \cite{GoldreichTrevisan} in their proof that such a
transformation exist for testing dense graphs. The ideas is to randomly pick a set of points,
$x_1,\ldots,x_d$, query $f$ on all points spanned by them and then decide according to the
answers. This is similar to the test of \cite{GoldreichTrevisan} which works by sampling a set of vertices
$v_1,\ldots,v_d$, querying about all edges spanned by them, and then answering according to these
answers.  We note that while in the case of graphs this results in a quadratic increase in the query
complexity (since $d$ vertices span $d^2$ edges) in our case the increase is exponential (since $d$
vectors over $\F^n_2$ span $2^d$) vector. We sketch the argument in Appendix ?????
}

\item As we have mentioned in Subsection \ref{subseccharac}, it is not too hard to construct
linear-invariant properties which are not testable. Actually, there are properties of this type that
cannot be tested with $o(2^n)$ queries. One example can be obtained from a variant of an
argument used in \cite{GGR} as follows; it is shown in \cite{GGR} (see Proposition 4.1) that for
every $n$ there exists a property of Boolean functions that contains $2^{\frac{1}{10}2^n}$ of the
Boolean functions over $\F_2^n$ and cannot be tested with less than $\frac{1}{20}2^n$ queries. This
family of functions is not necessarily linear invariant, so we just ``close'' it under linear
transformation, by adding to the property all the linear-transformed such functions. Since the
number of these linear transformation is bounded by $2^{n^2}$ (corresponding to all possible $n
\times n$ matrices over $\F_2$) we get that the new property contains at most
$2^{n^2}2^{\frac{1}{10}2^n} \leq 2^{\frac{1}{5}2^n}$ Boolean functions. One can verify that since
this new family contains a small fraction of all possible functions the argument of \cite{GGR}
caries over, and the new property cannot be tested with $o(2^n)$ queries.

\ignore{
\item
Our proof techniques actually show testability for a class of properties slightly larger than that
specified in Theorem \ref{thm:main}.  We can use Lemma 2.7 from \cite{BCSX09} to show that whenever
the linear system of equations described by the matrix $M$ is of complexity $1$ (see \cite{GT06} or
\cite{BCSX09} for definition), then our Lemma \ref{lem:count} still holds while the rest of the
proof machinery is unaffected.  For linear systems of larger complexity, bounds on higher-order
Gowers norms are needed to control the terms in the counting lemma.}

\item The upper bound one obtains from the general result given in Theorem \ref{thm:main} is
  terrible in terms of its dependence on $1/\eps$. A
natural open problem would be to find a characterization of these properties that can be tested
with a number of queries that depends polynomially on $\epsilon$. This, however, seems to be a very
hard problem. Even if the only forbidden equation is $x+y=z$ it is not known if such an efficient
test exists. This question was raised by Green \cite{Green05}; see \cite{BX10} for current best bounds.

\ignore{
\item Another natural question is if one can obtain any bound that would apply to all of the
properties covered by Theorem \ref{thm:main}, say, show that they can all be tested using
$2^{1/\epsilon}$ queries. It is natural to conjecture that since the family of properties considered
by Theorem \ref{thm:main} is rather general it would contain properties that require arbitrarily
large query complexity.  We observe that the proof of the analogous fact for graph properties \cite{ASmon}
readily extends to our setting also.  We omit the details here.
}

\item Our result here gives a (conjectured) characterization of the linear-invariant properties of
  Boolean functions that can be tested with one-sided error. It is of course natural to try to
  extend our framework to other families of properties, characterized by other or more general
  invariances.  For instance, can we carry out a full characterization for testable affine invariant
  properties of Boolean functions on the hypercube?

\item
It would be valuable to understand formally why the technology developed for
  handling graph properties can be extended so naturally to linear-invariant properties.
  This ``coincidence'' seems part of a larger trend in mathematics where claims about subsets find
  analogs in claims about vector subspaces.  See \cite{CohnF1} for an interesting attempt to shed
  light on this   puzzle.

\end{itemize}

\paragraph{Acknowledgements}
Arnab would like to thank Eldar Fischer for some initial stimulating discussions during a visit to
the Technion and Alex Samorodnitsky for constant encouragement and advice.  Many thanks also to Noga
Alon and Shachar Lovett for useful suggestions.

\bibliographystyle{alpha}
\bibliography{testing}

\newcommand{\etalchar}[1]{$^{#1}$}
\begin{thebibliography}{DLM{\etalchar{+}}07}

\bibitem[AFKS00]{AFKS}
Noga Alon, Eldar Fischer, Michael Krivelevich, and Mario Szegedy.
\newblock Efficient testing of large graphs.
\newblock {\em Combinatorica}, 20(4):451--476, 2000.

\bibitem[AFNS06]{AFNS06}
Noga Alon, Eldar Fischer, Ilan Newman, and Asaf Shapira.
\newblock A combinatorial characterization of the testable graph properties:
  it's all about regularity.
\newblock In {\em Proc.\ 36th Annual ACM Symposium on the Theory of Computing},
  pages 251--260, 2006.

\bibitem[AKK{\etalchar{+}}05]{AKKLR}
Noga Alon, Tali Kaufman, Michael Krivelevich, Simon Litsyn, and Dana Ron.
\newblock Testing {Reed-Muller} codes.
\newblock {\em IEEE Transactions on Information Theory}, 51(11):4032--4039,
  2005.

\bibitem[AS08a]{AS08}
Noga Alon and Asaf Shapira.
\newblock A characterization of the (natural) graph properties testable with
  one-sided error.
\newblock {\em SIAM J. on Comput.}, 37(6):1703--1727, 2008.

\bibitem[AS08b]{ASseparation}
Noga Alon and Asaf Shapira.
\newblock A separation theorem in property testing.
\newblock {\em Combinatorica}, 28:261--281, 2008.

\bibitem[AT08]{AT08}
Tim Austin and Terence Tao.
\newblock On the testability and repair of hereditary hypergraph properties.
\newblock {\em Random Structures and Algorithms (to appear)}, 2008.
\newblock Preprint available at \url{http://arxiv.org/abs/0801.2179}.

\bibitem[BCL{\etalchar{+}}06]{BCLSSV06}
Christian Borgs, Jennifer~T. Chayes, L{\'a}szl{\'o} Lov{\'a}sz, Vera~T.
  S{\'o}s, Bal{\'a}zs Szegedy, and Katalin Vesztergombi.
\newblock Graph limits and parameter testing.
\newblock In {\em Proc.\ 36th Annual ACM Symposium on the Theory of Computing},
  pages 261--270, 2006.

\bibitem[BCSX09]{BCSX09}
Arnab Bhattacharyya, Victor Chen, Madhu Sudan, and Ning Xie.
\newblock Testing linear-invariant non-linear properties.
\newblock In {\em STACS}, pages 135--146, 2009.
\newblock Full version at \url{http://www.eccc.uni-trier.de/report/2008/088/}.

\bibitem[BFL91]{BFL}
L{\'a}szl{\'o} Babai, Lance Fortnow, and Carsten Lund.
\newblock Non-deterministic exponential time has two-prover interactive
  protocols.
\newblock {\em Computational Complexity}, 1(1):3--40, 1991.

\bibitem[BHR05]{BHR05}
Eli {Ben-Sasson}, Prahladh Harsha, and Sofya Raskhodnikova.
\newblock Some 3cnf properties are hard to test.
\newblock {\em SIAM J. on Comput.}, 35(1):1--21, 2005.

\bibitem[BKS{\etalchar{+}}09]{BKSSZ}
Arnab Bhattacharyya, Swastik Kopparty, Grant Schoenebeck, Madhu Sudan, and
  David Zuckerman.
\newblock Optimal testing of {R}eed-{M}uller codes.
\newblock {\em Electronic Colloquium in Computational Complexity}, TR09-086,
  October 2009.

\bibitem[Bla09]{blais-junta}
Eric Blais.
\newblock Testing juntas nearly optimally.
\newblock In {\em Proc.\ 41st Annual ACM Symposium on the Theory of Computing},
  pages 151--158, 2009.

\bibitem[BLR93]{BLR}
Manuel Blum, Michael Luby, and Ronitt Rubinfeld.
\newblock Self-testing/correcting with applications to numerical problems.
\newblock {\em J. Comp. Sys. Sci.}, 47:549--595, 1993.
\newblock Earlier version in STOC'90.

\bibitem[BO10]{BO10}
Eric Blais and Ryan O'Donnell.
\newblock Lower bounds for testing function isomorphism.
\newblock In {\em Proc.\ 25th Annual IEEE Conference on Computational
  Complexity (to appear)}, 2010.

\bibitem[BS09]{BS09}
Eli {Ben-Sasson} and Madhu Sudan.
\newblock Limits on the rate of locally testable affine-invariant codes.
\newblock November 2009.
\newblock Manuscript.

\bibitem[BX10]{BX10}
Arnab Bhattacharyya and Ning Xie.
\newblock Lower bounds for testing triangle-freeness in boolean functions.
\newblock In {\em Proc.\ 21st {ACM}-{SIAM} {S}ymposium on {D}iscrete
  {A}lgorithms}, pages 87--98, Philadelphia, PA, USA, 2010. Society for
  Industrial and Applied Mathematics.

\bibitem[Coh04]{CohnF1}
Henry Cohn.
\newblock Projective geometry over ${\F}_1$ and the {G}aussian binomial
  coefficients.
\newblock {\em American Mathematical Monthly}, 111:487--495, 2004.

\bibitem[CRSW83]{CRST}
C.~Chvat\'al, V.~R\"odl, E.~Szemer\'edi, and {W. T. Trotter Jr.}
\newblock The {R}amsey number of a graph with bounded maximum degree.
\newblock {\em Journal of Combinatorial Theory, Series B}, 34(3):239--243,
  1983.

\bibitem[DLM{\etalchar{+}}07]{diakonikolasLMORSW07}
Ilias Diakonikolas, Homin~K. Lee, Kevin Matulef, Krzysztof Onak, Ronitt
  Rubinfeld, Rocco~A. Servedio, and Andrew Wan.
\newblock Testing for concise representations.
\newblock In {\em Proc.\ 48th Annual IEEE Symposium on Foundations of Computer
  Science}, pages 549--558, 2007.

\bibitem[Fis04]{FischerSurvey}
Eldar Fischer.
\newblock The art of uninformed decisions: {A} primer to property testing.
\newblock In G.~Paun, G.~Rozenberg, and A.~Salomaa, editors, {\em Current
  Trends in Theoretical Computer Science: The Challenge of the New Century},
  volume~1, pages 229--264. World Scientific Publishing, 2004.

\bibitem[Fis05]{FischerIsom}
Eldar Fischer.
\newblock The difficulty of testing for isomorphism against a graph that is
  given in advance.
\newblock {\em SIAM J. on Comput.}, 34(5):1147--1158, 2005.

\bibitem[FKR{\etalchar{+}}04]{FKRSS04}
Eldar Fischer, Guy Kindler, Dana Ron, Shmuel Safra, and Alex Samorodnitsky.
\newblock Testing juntas.
\newblock {\em J. Comp. Sys. Sci.}, 68(4):753--787, 2004.

\bibitem[GGR98]{GGR}
Oded Goldreich, Shafi Goldwasser, and Dana Ron.
\newblock Property testing and its connection to learning and approximation.
\newblock {\em Journal of the ACM}, 45:653--750, 1998.

\bibitem[GKS08]{GKS08}
Elena Grigorescu, Tali Kaufman, and Madhu Sudan.
\newblock 2-transitivity is insufficient for local testability.
\newblock In {\em IEEE Conference on Computational Complexity}, pages 259--267,
  2008.

\bibitem[GKS09]{GKS09}
Elena Grigorescu, Tali Kaufman, and Madhu Sudan.
\newblock Succinct representation of codes with applications to testing.
\newblock In {\em APPROX-RANDOM}, pages 534--547, 2009.

\bibitem[GOS{\etalchar{+}}09]{GOSSW}
Parikshit Gopalan, Ryan O'Donnell, Rocco~A. Servedio, Amir Shpilka, and Karl
  Wimmer.
\newblock Testing {F}ourier dimensionality and sparsity.
\newblock In {\em ICALP (1)}, pages 500--512, 2009.

\bibitem[Gow01]{Gow01}
William~T. Gowers.
\newblock A new proof of {S}zem\'eredi's theorem.
\newblock {\em Geometric Functional Analysis}, 11(3):465--588, 2001.

\bibitem[Gre05]{Green05}
Ben Green.
\newblock A {S}zemer\'edi-type regularity lemma in abelian groups.
\newblock {\em Geometric and Functional Analysis}, 15(2):340--376, 2005.

\bibitem[Gre10]{GreenPers}
Ben Green.
\newblock Personal communication, February 2010.

\bibitem[GT03]{GoldreichTrevisan}
Oded Goldreich and Luca Trevisan.
\newblock Three theorems regarding testing graph properties.
\newblock {\em Random Structures and Algorithms}, 23(1):23--57, 2003.

\bibitem[GT08]{GT06}
Ben Green and Terence Tao.
\newblock Linear equations in primes.
\newblock Preprint available at \url{http://arxiv.org/abs/math/0606088v2},
  April 2008.

\bibitem[GT10]{GT10}
Ben Green and Terence Tao.
\newblock An arithmetic regularity lemma, associated counting lemma, and
  applications.
\newblock Preprint available at \url{http://arxiv.org/abs/1002.2028}, February
  2010.

\bibitem[KL05]{KaufmanLitsyn}
Tali Kaufman and Simon Litsyn.
\newblock Almost orthogonal linear codes are locally testable.
\newblock {\em Proc.\ 46th Annual IEEE Symposium on Foundations of Computer
  Science}, pages 317--326, 2005.

\bibitem[KS07]{KS07}
Tali Kaufman and Madhu Sudan.
\newblock Sparse random linear codes are locally decodable and testable.
\newblock In {\em Proc.\ 48th Annual IEEE Symposium on Foundations of Computer
  Science}, pages 590--600, 2007.

\bibitem[KS08]{KS}
Tali Kaufman and Madhu Sudan.
\newblock Algebraic property testing: the role of invariance.
\newblock In {\em Proc.\ 40th Annual ACM Symposium on the Theory of Computing},
  pages 403--412, New York, NY, USA, 2008. ACM.

\bibitem[KS09]{KoppartyS09}
Swastik Kopparty and Shubhangi Saraf.
\newblock Tolerant linearity testing and locally testable codes.
\newblock In {\em APPROX-RANDOM}, pages 601--614, 2009.

\bibitem[KSV08]{KSV08b}
Daniel Kr\'al{'}, Oriol Serra, and Llu\'{i}s Vena.
\newblock A removal lemma for systems of linear equations over finite fields.
\newblock {\em Israel Journal of Mathematics (to appear)}, 2008.
\newblock Preprint available at \url{http://arxiv.org/abs/0809.1846}.

\bibitem[KSV09]{KSV09}
Daniel Kr\'{a}l{'}, Oriol Serra, and Llu\'{i}s Vena.
\newblock A combinatorial proof of the removal lemma for groups.
\newblock {\em Journal of Combinatorial Theory}, 116(4):971--978, May 2009.

\bibitem[MORS09]{MatulefORS09}
Kevin Matulef, Ryan O'Donnell, Ronitt Rubinfeld, and Rocco~A. Servedio.
\newblock Testing halfspaces.
\newblock In {\em Proc.\ 20th {ACM}-{SIAM} {S}ymposium on {D}iscrete
  {A}lgorithms}, pages 256--264, 2009.

\bibitem[PRS02]{ParnasRS02}
Michal Parnas, Dana Ron, and Alex Samorodnitsky.
\newblock Testing basic boolean formulae.
\newblock {\em SIAM J. Discrete Math.}, 16(1):20--46, 2002.

\bibitem[Ron08]{RonSurvey}
Dana Ron.
\newblock {P}roperty {T}esting: {A} {L}earning {T}heory {P}erspective.
\newblock In {\em Foundations and Trends in Machine Learning}, volume~1, pages
  307--402. 2008.

\bibitem[RS96]{RS}
Ronitt Rubinfeld and Madhu Sudan.
\newblock Robust characterizations of polynomials with applications to program
  testing.
\newblock {\em SIAM J. on Comput.}, 25:252--271, 1996.

\bibitem[RS09]{RodlSchacht}
Vojt\v{e}ch R\"odl and Mathias Schacht.
\newblock Generalizations of the removal lemma.
\newblock {\em Combinatorica}, 29(4):467--502, 2009.

\bibitem[Rub06]{RubinfeldICM}
Ronitt Rubinfeld.
\newblock Sublinear time algorithms.
\newblock In {\em Proceedings of International Congress of Mathematicians
  2006}, volume~3, pages 1095--1110, 2006.

\bibitem[Sam07]{Samorodnitsky07}
Alex Samorodnitsky.
\newblock Low-degree tests at large distances.
\newblock In {\em Proc.\ 37th Annual ACM Symposium on the Theory of Computing},
  pages 506--515, 2007.

\bibitem[Sha09]{Shap09}
Asaf Shapira.
\newblock Green's conjecture and testing linear-invariant properties.
\newblock In {\em Proc.\ 41st Annual ACM Symposium on the Theory of Computing},
  pages 159--166, 2009.

\bibitem[Sud10]{SudanSurvey}
Madhu Sudan.
\newblock Invariance in property testing.
\newblock {\em Electronic Colloquium in Computational Complexity}, TR10-051,
  March 2010.

\bibitem[Sze78]{Szemeredi}
Endre Szemer\'edi.
\newblock Regular partitions of graphs.
\newblock In J.C. Bremond, J.C. Fournier, M.~Las Vergnas, and D.~Sotteau,
  editors, {\em Proc. Colloque Internationaux CNRS 260 -— Probl\`emes
  Combinatoires et Th\'eorie des Graphes}, pages 399--401, 1978.

\bibitem[Sze10]{SzegSym10}
Bal\'azs Szegedy.
\newblock The symmetry preserving regularity lemma.
\newblock {\em Proc. Amer. Math. Soc.}, 138:405--408, 2010.

\end{thebibliography}


\appendix

\section{Proofs omitted from Section \ref{sec:intro}}\label{app:reedmuller}

\begin{paragraph}{Characterization of Reed Muller codes by forbidding systems of induced equations}
 First recall that Reed Muller codes of order $d$ are defined as
$${\cal R}{\cal M}(d)=\left\{ f:\F_2^n \rightarrow \F_2 :~ f(x)=\sum_{S\subset [n], |S|\leq d}~ \prod_{i\in S}~ x_i.\right\}$$

The most common characterization of ${\cal R}{\cal M}(d)$ (see for example \cite{AKKLR}) is that $f\in{\cal R}{\cal M}(d)$  if and only if $f$ satisfies
$$ \sum_{S\subset [n], |S|\leq d+1 } f\left(\alpha +\sum_{i\in S} \alpha_i\right)=0, \mbox{ for all } (\alpha, \alpha_1, \ldots, \alpha_{d+1})\in (\F_2^n)^{d+2}. $$

We use this description to  obtain a matrix $M\in \F_2^{(2^{d+1}-d-2)\times (2^{d+1})}$ and a collection of $\sigma^i\in \{0,1\}^{2^{d+1}}$ such that ${\cal R}{\cal M}(d)$ is $\{(M, \sigma^i)\}_i$- free. Intuitively, we want $M$ to encode all the linear relations between the elements of the set $A=\{\alpha+\sum_{i\in S} \alpha_i \}_{0\leq |S|\leq d+1}$, and we want to use the $\sigma^i$'s to enforce the fact $f$ should evaluate to $1$ on an even number of elements of $A$.

More exactly, assume that $B=\{\alpha, \alpha+\alpha_1, \ldots, \alpha+\alpha_{d+1}\}$ are linearly independent.
For every   $\beta\in A-B$, add to $M$ the row which is the vector representing  $\beta$ in the basis $B$. Further, consider all the $\sigma^i \in \{0, 1\}^{2^{d+1}}$ such that $|\{j: \sigma^i_j=1\}|$ is odd. Clearly the number of such $\sigma^i$'s is finite, and the patterns allowed by forbidding all $(M, \sigma^i)$ are only those that satisfy the above characterization.

Finally, notice that setting $d=1$ the resulting matrix $M$ contains only one row, and thus  Theorem~\ref{thm:main} applies to testing linearity.
\end{paragraph}

\bigskip
 We conclude with the proof
of Proposition \ref{HereditaryInduced} which was also omitted from the Introduction.

\bigskip

\begin{proofof}{Proposition \ref{HereditaryInduced}}
In one direction, it is easy to check that $\calf$-freeness is a subspace-hereditary linear-invariant
property, for any fixed family $\calf$.

Now, we show the other direction.  For a subspace-hereditary linear-invariant property $\mathcal{P}$, let
$\mathsf{Obs}$ denote the collection of pairs $(d,S)$, where $d \geq 1$ is an integer and $S
\subseteq \F_2^d$ is a subset, such that $\indic_S$ does not have property $\mathcal{P}$ and is
minimal with respect to restriction to subspaces. In other words, $(d,S)$ is contained in
$\mathsf{Obs}$ iff $\indic_S \not \in \calp_d$ but for any vector subspace $U \subseteq \F_2^d$ of dimension
$d' < d$, $\indic_{S|_U} \in \mathcal{P}_{d'}$ where $S|_U \subseteq U$ is the restriction of $S$ to $U$.

For every $(d,S) \in \mathsf{Obs}$,  we construct a matrix $M_d$ and a tuple $\sigma_S$ such that any
$f$ with property $\calp$ is $(M_d,\sigma_S)$-free.  Define $A_d$ to be the $2^d$-by-$d$
matrix over $\F_2$, where each of the $2^d$ rows corresponds to a distinct element of $\F_2^d$
represented using some choice of bases.  Now, define $M_d$ to be a $(q^d-d)$-by-$q^d$ matrix over
$\F$, such that $M_d A_d = 0$ and $\mathsf{rank}(M_d) = q^d-d$.  Define $\sigma_S$ as
$(\sigma(1),\sigma(2),\dots,\sigma(2^d))$ where $\sigma(i) = \indic_S(x_i)$ with $x_i$ being the
element of $\F_2^d$ represented in the $i$th row of  $A_d$.  
We observe now that any $f :  \F_2^n \to \zo$ having property $\calp$ is $(M_d,  \sigma_S)$-free.
Suppose the opposite, so that there exists $x = (x_1,\dots,x_{q^d}) \in (\F_2^n)^d$ 
satisfying $Mx = 0$ and $f(x_i) = \sigma(i)$.  Then, by definition of $M_d$, the $x_1,\dots,
x_{2^d}$ are the elements of a $d$-dimensional subspace $V$ over $\F_2$, and by definition of
$\sigma_S$, $S_f|_ V = S$ where $S_f$ is the support of $f$.  Thus $f|_V \not \in \calp$ which is a
contradiction to the fact that $f$ has property $\calp$ because $\calp$ is subspace-hereditary.  

Finally, define $\calf_\calp = \{(M_d, \sigma_S)\}$.  We have just seen that any $f$ having property
$\calp$ is $\calf_\calp$-free.  On the other hand, suppose $f$ does not have property $\calp$. Then,
because of heredity, there must be a $d$-dimensional subspace $V$ such that the support of $f|_V$ is
isomorphic to $S$ for some $(d,S) \in \mathsf{Obs}$ under linear transformations, which means by the
same argument as above, that $f$ will not be $(M_d,\sigma_S)$-free.
\ignore{
In the case that $\calp$ is known to be monotone-hereditary, we can construct $\calf_\calp$ as
follows.  Let $\mathsf{Obs}'$ denote the collection of sets which do not have property $\calp$ and
which are both minimal with respect to restriction to subspaces as well as removal of elements.
That is, here, $S$ is contained in $\mathsf{Obs}'$ iff there is a $d \geq 1$ such that $S
\subseteq \F^d$ and $S \not \in \calp_d$, but $(i)$, for any vector subspace $U \subseteq \F^d$ of
dimension $d' < d$, $S|_U \in \mathcal{P}_{d'}$, and $(ii)$, for any $S' \subsetneq S$, $S'
\in \calp_d$.  For every $S \in \mathsf{Obs}'$, we construct a matrix $M_S$. Define $A_S$ to be the
$|S|$-by-$d$ matrix over $\F$ where each row corresponds to
a distinct element of $S$ represented using some choice of bases, and let $M_S$ be a
$(|S|-\mathsf{rank}(A_S))$-by-$|S|$ matrix over $\F$ such that $M_SA_S = 0$ and $\mathsf{rank}(M_S)
= |S| - \mathsf{rank}(A_S)$.  Let $\calf_\calp = \{M_S : S \in \mathsf{Obs}'\}$.  A similar argument
to above shows that $\calf_\calp$-freeness is equivalent to $\calp$.}
\end{proofof}



\end{document}